\documentclass[11pt]{article}
\usepackage[arxiv]{mymacros}

\title{Spectral gap critical exponent for Glauber dynamics \\ of hierarchical spin models}
\author{Roland Bauerschmidt\footnote{University of Cambridge, Statistical Laboratory, DPMMS. E-mail: {\tt rb812@cam.ac.uk}.} \and
  \and Thierry Bodineau\thanks{CMAP \'Ecole Polytechnique, CNRS, Universit\'e Paris-Saclay. E-mail: {\tt thierry.bodineau@polytechnique.edu}.}}
\date{July 10, 2019}

\begin{document}
\maketitle
\begin{abstract}
  We develop a renormalisation group approach to deriving the asymptotics
  of the spectral gap of the generator of Glauber type dynamics of spin systems
  with strong correlations (at and near a critical point).
  In our approach, we derive a spectral gap inequality 
  for the measure recursively in terms of spectral gap 
  inequalities for a sequence of renormalised measures.
  We apply our method to hierarchical versions of
  the $4$-dimensional $n$-component $|\varphi|^4$ model
  at the critical point and its approach from the high temperature side,
  and of the $2$-dimensional Sine-Gordon and the Discrete Gaussian models
  in the rough phase (Kosterlitz--Thouless phase).
  For these models, we show that the spectral gap decays polynomially
  like the spectral gap of the dynamics of a free field (with a logarithmic correction for the $|\varphi|^4$ model),
  the scaling limit of these models in equilibrium.
\end{abstract}

\section{Introduction and main results}
\label{sec:intro}

\subsection{Introduction}

Spin systems in equilibrium have been studied by a variety of methods
which led to a very complete mathematical description of the physical phenomena occurring in the different regimes of the phase diagrams. 
This includes in particular
a good understanding of the critical phenomena in a wide range of models.  
Much less is known about the Glauber dynamics of spin systems.
For sufficiently high temperatures, it is well understood that the dynamics relaxes exponentially fast towards the equilibrium measure.
For the Ising model, the much more difficult question of fast relaxation in the entire uniqueness regime was addressed in 
\cite{MR1269388,MR1269387,MR3844472,MR3486171}.
In the phase transition regime, at least for scalar spins, the dynamical behaviour is governed by the interface motion and 
the relaxation becomes much slower. 
In particular, the relaxation time diverges  as  the system size increases, 
but the dynamical scaling depends strongly on the choice of the boundary conditions. We refer to  \cite{MR1746301}
for a review, as well as to \cite{MR2663712,MR3017041} for more recent results.
In the vicinity of the critical point, strong correlations develop and as a consequence the dynamic
evolution slows down but is no longer driven by phase separation. 
Even though the critical dynamical behaviour has been well investigated in physics
\cite{halperinhohenberg1977}, mathematical results are scarce. 
The only cases for which polynomial lower bounds on the relaxation or mixing times are known are  
the two-dimensional Ising model \cite{MR2945623}, exactly at the critical point,
the Ising model on a tree \cite{MR2585995}, both without sharp exponent, and
the mean-field Ising model which is fully understood \cite{MR2550363,MR2506768}.

\medskip

The goal of this paper is to investigate the dynamical relaxation of  hierarchical models near and at the critical point by deriving the scaling of the spectral gap in terms of
the temperature (or the equivalent parameter of the model) and the system size.

Since their introduction by Dyson \cite{MR0436850} and the pioneering work of Bleher--Sinai \cite{MR1552598},
hierarchical models have been a stepping stone to develop renormalisation group arguments.
At equilibrium, sharp results on the critical behaviour of a large class of models have typically been
obtained first in a hierarchical framework and then later been extended to the Euclidean lattice.
For the equilibrium problem, the hierarchical framework results in a significant technical simplification,
but the results and methods have turned out to be surprisingly parallel to the case of the Euclidean lattice $\Z^d$.
This point of view is discussed in detail in \cite{rg-brief},
to which we also refer for an overview of results and references.
Building on the results for the hierarchical set-up for the equilibrium problem,
we  derive recursive relations on the spectral gap after one renormalisation step. 
This enables us to obtain sharp asymptotic behaviour  of the  spectral gap for
large size Sine-Gordon model in the rough phase (Kosterlitz--Thouless phase)
and for the $|\varphi|^4$ model in the vicinity of the critical point. 
The scaling coincides in both cases with the one of the  hierarchical free field 
dynamics (with a logarithmic corrections for the $|\varphi|^4$ model)
which describes the equilibrium scaling limit of these models.
Renormalisation procedures have already been used to analyze spectral gaps for Glauber dynamics,
see e.g., \cite{MR1746301}, but the renormalisation scheme used in this paper is different
and allows to keep sharp control from one  scale to the next.

\medskip

After recalling the definitions of the hierarchical models and presenting the results of this paper in Section \ref{sec: Models and results}, we implement, in Section \ref{sec:recursion}, 
the induction procedure to control the spectral gap after one renormalisation step. 
We believe that our method could be extended beyond the hierarchical models, thus the induction is described in a general framework under some assumptions which can then be checked for each microscopic models.
This is completed in Section~\ref{sec:phi4} for the hierarchical $|\varphi|^4$ model,
and in Section~\ref{sec:sg} for the hierarchical Sine-Gordon and the Discrete Gaussian models.
Proving these assumptions requires establishing
stronger control on the renormalised Hamiltonians in the \emph{large field region}
than needed when studying the renormalisation at equilibrium
(convexity instead of probabilistic bounds).
Such convexity for large fields is the main challenge to extend the method of this paper
beyond hierarchical models.

\subsection{Spectral gap}
\label{sec:intro-gap}

Let $\Lambda$ be a finite set and $M$ be a symmetric matrix of spin couplings acting on $\R^\Lambda$.
We consider possibly vector-valued spin configurations $\varphi = (\varphi_x^i)_{x\in\Lambda, i=1,\dots, n} \in \R^{n\Lambda} = \{ \varphi: \Lambda \to \R^n\}$,
with action of the form
\begin{equation} \label{e:Hdef}
  H(\varphi) = \frac12(\varphi,M\varphi) + \sum_{x\in\Lambda} V(\varphi_x), \quad (\varphi \in \R^{n\Lambda}),
\end{equation}
for some potential $V: \R^n\to\R$, where $(\cdot,\cdot)$ is the standard inner product on $\R^{n\Lambda}$.
In the vector-valued case $n>1$, we assume that $V$ is $O(n)$-invariant and that
$M$ acts by $(M\varphi)_x^i = (M\varphi^i)_x$ for $i=1,\dots, n$ and $x\in\Lambda$.
The associated probability measure $\mu$ has expectation
\begin{equation}
\label{eq: mu mesure}
  \bbE_\mu (F) = \frac{1}{Z} \int_{\R^{n\Lambda}} e^{-H(\varphi)} F(\varphi) \, d\varphi, \qquad
  Z = \int_{\R^{n\Lambda}} e^{-H(\varphi)} \, d\varphi.
\end{equation}
The (continuous) Glauber dynamics associated with  $H$ is given by the system of stochastic differential equations
\begin{equation} \label{e:Langevin}
  d\varphi_x = -\partial_{\varphi_x} H(\varphi) \, dt +  \sqrt{2} dB_x, \quad (x\in\Lambda),
\end{equation}
where the $B_x$ are independent $n$-dimensional standard Brownian motions.
(The continuous Glauber dynamics is also referred to as overdamped Langevin dynamics;
to keep the terminology concise we  use the term Glauber dynamics in the continuous as well as in the discrete case.)
By construction, the measure $\mu$ defined in \eqref{eq: mu mesure} is invariant with respect to this dynamics.
Its relaxation time scale is controlled by the inverse of the spectral gap of the generator of the Glauber dynamics
(see, for example, \cite[Proposition~2.1]{MR2213477}).
By definition, the spectral gap is the largest constant $\gamma$ such that,
for all functions $F: \R^{n\Lambda} \to \R$ with bounded derivative,
\begin{eqnarray}
\label{e:gap}
  \var_\mu(F)= \bbE_\mu (F^2) - \bbE_\mu(F)^2
  \leq 
  \frac{1}{\gamma} \bbE_\mu (\nabla F, \nabla F)  \, .
\end{eqnarray}

Our goal in this paper is to determine the order of the spectral gap $\gamma$ for specific choices of $M$ and $V$,
when the size of the domain $\gL$ diverges.
For statistical mechanics,
the setting of primary interest is a finite domain of a lattice or a
torus $\Lambda = \Lambda_N \subset \Z^d$ whose size tends to infinity,
and a short-range spin coupling matrix $M$, such as the discrete Laplace operator $-\Delta$ on $\Lambda$.
The discrete Laplace operator has a nontrivial kernel.
This degeneracy must be removed through boundary conditions or an external field (mass term).
For example, for a cube of side length $D$ with Dirichlet boundary conditions,
the smallest eigenvalue is of order $D^{-2}$.
In the hierarchical set-up that we consider,
we impose an external field instead of boundary conditions whose size is such that the
smallest eigenvalue is at least of order $D^{-2}$.

For $V=0$, or more generally for quadratic potentials which can be absorbed in the definition of $M$,
the spectral gap $\gamma$ of the generator of the Langevin dynamics is equal to
the minimal eigenvalue of $M$ (assuming that it is positive) 
by explicit diagonalisation of \eqref{e:Langevin}.
More generally, for $V$ any strictly convex potential
satisfying $V''(\varphi) \geq c > 0$ uniformly in $\varphi$,
the Bakry--Emery criterion \cite{MR889476} implies that
\begin{equation}
  \gamma \geq \lambda + c,
\end{equation}
where $\lambda$ is the smallest eigenvalue of $M$.
Under these conditions, $\mu$ actually satisfies a logarithmic Sobolev inequality with the same constant.
In particular, under these assumptions, the dynamics relaxes quickly, in time of order $1$.

The situation is much more subtle when the potential $V$ is non-convex.
Indeed, as the potential becomes sufficiently non-convex,
the static measure $\mu$ typically undergoes phase transitions.
In fact for unbounded spin systems on a lattice, the relaxation of the Glauber dynamics has been
controlled only in the uniqueness regime under some assumptions  on the decay of correlations 
\cite{MR1715549, MR1936110, MR1764741, MR1704666, MR1837286} 
(see also \cite{MR3098070} for  conservative dynamics).
By considering hierarchical models, we are able to show that the spectral gap decays
polynomially in the vicinity of a phase transition.
The idea is to decompose the measure into renormalised fields such that at each scale,
conditioned on a block spin field, the renormalised potential remains strictly convex.
By induction, we then obtain a recursion on the spectral gaps of the renormalised measures.

Before stating the results,  we first turn to the definition of the hierarchical models.

\subsection{Hierarchical Laplacian}
\label{sec:hierGFF}

The Gaussian free field (GFF) on a finite approximation to $\Z^d$ is a Gaussian field whose
covariance is the Green function of the Laplace operator. The Green function
has decay $|x|^{-(d-2)}$ in dimensions $d\geq 3$ and has asymptotic behaviour $-\log |x|$ in dimension $d=2$.
The hierarchical Laplace operator is an approximation to the Euclidean one in the sense that its Green function
has comparable long-distance behaviour, but simpler short-distance structure.
The study of hierarchical models has a long history in statistical mechanics going back to
\cite{MR0436850,MR1552598}; recent studies and uses of hierarchical models include
\cite{MR880526,MR1143413,MR3422923,MR3526836,1302.5971} and references.

\begin{figure}
\begin{center}
  \input{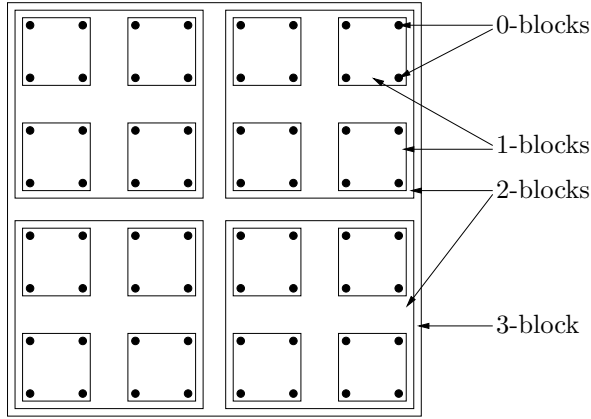}
\end{center}
\caption{Blocks in ${\cal B}_j$ for $j=0,1,2,3$ where $d=2$, $N=3$, $L=2$.
  \label{fig:hier1}}
\end{figure}

There is some flexibility in the choice of the hierarchical field; the precise choice is not significant.
Let $\Lambda = \Lambda_N$ be a cube of side length $L^N$ in $\Z^d$, $d \geq 1$, for some fixed integer $L>1$ and $N$ eventually chosen large.
For scale $0\leq j \leq N$, we decompose $\Lambda$ as the union of disjoint blocks of side lengths $L^j$ denoted $B \in \mathcal{B}_j$;
see Figure~\ref{fig:hier1}.
In particular, $\mathcal{B}_0 = \Lambda$ and the unique block in $\mathcal{B}_N$ is $\Lambda_N$ itself.
The blocks have the structure of a $K$-ary tree with $K=L^d$,
height $N$ and the leaves are indexed by the sites $x \in \Lambda_N$.

For scale $j$ and $x\in\Lambda$, let $B_{j}(x)$ be the block in $\mathcal{B}_j$ containing $x$.
As in \cite[Chapter 4]{rg-brief},
define the block averaging operators, which are the projections
\begin{equation} \label{e:QPdef}
  (Q_jf)_x = \frac{1}{|B_{j}(x)|}\sum_{y \in B_{j}(x)} f_y, \quad \text{for $f\in \R^\Lambda$.}
\end{equation}
Let $P_j = Q_{j-1}-Q_{j}$.
Then $P_1, \dots, P_N, Q_N$ are orthogonal projections on $\R^\Lambda$ with disjoint ranges whose direct sum
is the full space.
An operator on $\R^\Lambda$ is \emph{hierarchical} if it is diagonal with respect to this decomposition.
To obtain a hierarchical Green function with the scaling of the Green function of the usual Laplace operator,
we choose the hierarchical Laplace operator on $\Lambda$ to be
\begin{equation} 
\label{e:DeltaH}
-\Delta_{H}
= \sum_{j=1}^N
L^{-2(j-1)} P_j
.
\end{equation}
Like the usual Laplacian on the discrete torus, this choice of hierarchical Laplacian annihilates the constant functions.
The definition implies that the Green function of the hierarchical Laplacian has comparable long distance behaviour
to that of the nearest-neighbour Laplacian: for $|x-y|^{-1} \ll m$,
\begin{alignat}{2}
  \label{e:DeltaH-asymp1}
  (-\Delta_H+m^2)^{-1}_{xy} &\asymp |x-y|^{-(d-2)} &&\qquad (d>2),
  \\
  \label{e:DeltaH-asymp2}
  (-\Delta_H+m^2)^{-1}_{xy} &= c_N - \sigma \log_L |x-y| + O(1) &&\qquad (d=2),
\end{alignat}
where $|x-y|$ is the Euclidean distance and $\sigma = 1-L^{-2}$ is a constant independent of $N$,
and $A \asymp B$ denotes that $A/B$ and $B/A$ are bounded by $N$-independent constants.
On the other hand, the hierarchical Laplacian has coarser small distance behaviour
than the lattice Laplacian.
For a more detailed introduction to the hierarchical Laplacian,
as well as discussion of its relation to the lattice Laplacian,
see \cite[Chapters~3--4]{rg-brief}.

\subsection{Models and results}
\label{sec: Models and results}

In Section~\ref{sec:recursion},
we are going to develop a quite general multiscale strategy to estimate the spectral gap of
(critical) spin systems by using a renormalisation  group approach.
We will then apply this method to
the $n$-component $|\varphi|^4$ model and the Sine-Gordon model
as well as the degenerate case of the Discrete Gaussian model.
These models correspond to choices of the potential~$V$ defined now.
In the setting of the hierarchical spin coupling, we study
the critical region of the $|\varphi|^4$ model and the rough phase of the Sine-Gordon and Discrete Gaussian models.
These are both settings for which the renormalisation group
method is well developed for the equilibrium case, and we use this as input.

\subsubsection{Ginzburg--Landau--Wilson $|\varphi|^4$ model}

The $n$-component $|\varphi|^4$ model is defined by the double-well potential (if $n=1$),
respectively Mexican hat shaped potential (if $n\geq 2$),
\begin{equation} \label{e:potential-phi4}
  M=-\Delta_H,
  \quad V(\varphi) = \frac14 g|\varphi|^4 + \frac12 \nu|\varphi|^2, \quad (g>0, \; \nu \in \R).
\end{equation}
Our interest is in the case $\nu<0$, when this potential is non-convex.
The $|\varphi|^4$ model is a prototype for a spin model with $O(n)$ symmetry.
The spatial dimension $d=4$ is critical for this model (see, e.g., \cite{rg-brief}).
The following theorem quantifies the decay of the spectral gap in the four-dimensional
hierarchical $|\varphi|^4$ model 
when approaching the critical point from the high temperature side.

\begin{theorem} 
\label{thm:gap-phi4}
Let $\gamma_N(g,\nu,n)$ be the spectral gap of the hierarchical $n$-component $|\varphi|^4$ model on $\Lambda_N$
with dimension $d=4$ (as defined above).
Let $L \geq L_0$, and let $g>0$ be sufficiently small.
There exists $\nu_c = \nu_c(g,n) = -C(n+2)g + O(g^2)$ and a constant $\delta \geq 1$ (independent of $n$)
such that for $t_0 \geq t \geq cL^{-2N}$, where $t_0$ is a small constant,
\begin{equation} \label{e:gap-phi4}
  c_1 t (-\log t)^{-\delta (n+2)/(n+8)}
  \leq
  \gamma_N(g,\nu_c+t,n)
  \leq
  c_2 t (-\log t)^{-(n+2)/(n+8)},
\end{equation}
provided that $N$ is sufficiently large. In particular, $t\geq cL^{-2N}$ is allowed to depend on $N$.
\end{theorem}

The proof is postponed to Section \ref{sec:phi4}.
The same proof also implies easily that for $t \geq t_0$ the gap is of order $1$, but since
we are interested in the more delicate approach of the critical point, we omit the details.
Together with this, Theorem \ref{thm:gap-phi4} implies that 
for the $|\varphi|^4$ model,
the spectral gap is of order $1$ in the high temperature phase, $\nu > \nu_c$ independently of $N$,
and as the critical point is approached the spectral gap scales like that of the free field,
with a logarithmic correction.
We expect that $\gamma \sim Ct(-\log t)^{-z}$
for a universal critical exponent $z = z(n) \geq \frac{n+2}{n+8}$, which our method does not determine
(see also \cite{halperinhohenberg1977}).
The upper bound 
follows easily from the estimates derived at equilibrium in 
\cite[Theorem~4.2.1]{rg-brief} and we also use the renormalisation group
flow constructed in \cite{rg-brief} as input to prove the lower bound
(see also \cite{MR693402}).
References for the renormalisation group analysis of the $|\varphi|^4$ model on $\Z^4$,
with different approaches,
include \cite{MR790736,MR892924,MR892925}, \cite{MR882810} and
\cite{MR3332938,MR3332939,MR3332940,MR3332941,MR3332942,MR3345374,MR3339164,MR3269689}.

\subsubsection{Sine-Gordon model}

The Sine-Gordon model is defined by a $2\pi$-periodic potential and coupling matrix proportional to the inverse temperature $\beta$, i.e.,
\begin{equation} \label{e:potential-sinegordon}
  M=-\beta\Delta_H   \quad (\beta>0), \qquad
  \text{$V(\varphi)$ is even and
    $2\pi$-periodic}.
\end{equation}

The corresponding energy $H(\varphi)$ in \eqref{e:Hdef} is invariant under
$\varphi \mapsto \varphi + 2\pi n \b 1$ for any $n \in \Z$,
where $\b 1$ denotes the constant function on $\Lambda$ with $\b 1_x = 1$ for all $x\in\Lambda$.
To break this non-compact symmetry, we 
add the external field and consider
\begin{equation}
  H_\epsilon(\varphi)
  =
  H(\varphi)
  + \frac{\epsilon}{2} \pa{\frac{1}{\sqrt{|\Lambda|}}\sum_{x} \varphi_x}^2
  =
  \frac{\beta}{2} (\varphi,-\Delta_H \varphi) + \sum_x V(\varphi_x)
  + \frac{\epsilon}{2} \pa{\frac{1}{\sqrt{|\Lambda|}}\sum_{x} \varphi_x}^2
  .
\end{equation}
As previously, we are interested in the large volume limit $|\Lambda|\uparrow \infty$;
to avoid some uninteresting technicalities, we will make the convenient choice $\epsilon =\beta L^{-2N}$.
If $V$ was, e.g., the double well potential $V(\varphi) = \varphi^4-\varphi^2$ instead of a periodic potential as above,
then the corresponding measure has a uniform spectral gap for any $\beta>0$ sufficiently small (see, e.g., \cite{MR3926125}).
The following theorem shows that this is not the case for periodic potentials: the spectral gap decreases to $0$.
Thus that the resulting models are critical, in the sense of slow decay of correlations, is also reflected in their dynamics.

For the statement of the theorem,
denote by $\hat V(q) = (2\pi)^{-1}  \int_{-\pi}^\pi e^{iq\varphi} V(\varphi) \, d\varphi$
the Fourier coefficient of the $2\pi$-periodic function $V$,
and let $\sigma = 1-L^{-2}$ be the constant in \eqref{e:DeltaH-asymp2} with dimension $d=2$.

\begin{theorem} 
\label{thm:gap-sg}
Let $\gamma_N(\beta,V)$ be the spectral gap of the hierarchical Sine-Gordon model on $\Lambda_N$ with dimension $d=2$ (as defined above).
Assume $\sum_{q\in \Z \setminus \{0\}} (1+q^2) |\hat V(q)|$ is  small enough.
Let $0 < \beta < \sigma/(4\log L)$ and let $\epsilon = \beta L^{-2N}$.
There are $\kappa\in (0,1)$ and $c>0$ such that the spectral gap scales as
\begin{equation}
  cL^{-2N}
  \leq
  \gamma_N(\beta, V)
  \leq
  L^{-2N}(1-O(\kappa^N))
\end{equation}
provided that $N$ is sufficiently large.
\end{theorem}

The Sine-Gordon model is dual to a Coulomb gas model (see, e.g., \cite{MR2523458,MR634447}).
Under this duality, the inverse temperature of the Coulomb gas model is proportional
to the temperature $1/\beta$ of the Sine-Gordon model. We here primarily view 
the Sine-Gordon model as a spin model, rather than as a description of the Coulomb gas,
and therefore choose $\beta$ instead of $1/\beta$ in \eqref{e:potential-sinegordon}.
Note that the usual normalisation of the logarithm in \eqref{e:DeltaH-asymp2} is
$c_N - \frac{1}{2\pi} \log |x| + O(1)$
for the Laplace operator on $\Z^2$.
For this normalisation of the hierarchical Laplace operator, the hierarchical critical inverse
temperature becomes $1/\beta =  8\pi$.
This is only approximately true
in the Euclidean model because of a field-strength (stiffness) renormalisation
which is not present in the hierarchical model.
For the critical inverse temperature
$\beta = \sigma/(4\log L)$, we expect that $\gamma \sim C L^{-2N} N^{-z}$
for a universal critical exponent $z>0$. For the presence of logarithmic corrections to the free field scaling in the static case, see \cite{1311.2237}.
Our theorem uses the set-up for the renormalisation group for this model of \cite{MR2523458}
(see also \cite{MR1003504}).
References for the Sine-Gordon model on $\Z^2$ include \cite{MR634447} and \cite{MR2917175,1311.2237,
MR1063215,MR1777310,MR1240586,MR1101688}.

\subsubsection{Discrete Gaussian model}

We conclude this section with a discrete model which is closely linked to the Sine-Gordon model.
The Discrete Gaussian model is an integer-valued field with expectation given by
\begin{equation}
  \bbE_ \mu(F) = \frac{1}{Z} \sum_{\sigma \in (2\pi\Z)^\Lambda} F(\sigma) 
  e^{-\frac{\beta}{2}(\sigma,-\Delta_H\sigma) - \frac{\epsilon}{2} \p{\frac{1}{\sqrt{|\Lambda|}}\sum_{x}\sigma_x}^2}
  \quad
  \text{for $F: (2\pi\Z)^\Lambda \to \R$,}
  \qquad (\beta>0).
\end{equation}
Note that by rescaling $\beta$ and $\epsilon$ by $(2\pi)^2$, this definition is equivalent
to the one in which the model takes values in $\Z$ rather than $2\pi \Z$. The normalisation
by $2\pi$ is convenient for our proof.
The model formally takes the form of a degenerate Sine-Gordon model in which $e^{-V(\varphi)}$
is replaced by a sum of $\delta$-functions.
As the spins take integer values,
we now consider a discrete Glauber dynamics for the Discrete Gaussian model
with Dirichlet form
\begin{equation}
\label{eq: dirichlet discrete}
\frac{1}{2(2\pi)^2}
\sum_{x\in\Lambda}  \bbE_\mu  
\Big( (F(\sigma^{x+})-F(\sigma))^2  + (F(\sigma^{x-})-F(\sigma))^2 \Big),
\end{equation}
where $\sigma^{x\pm}$ is obtained from $\sigma \in (2\pi\Z)^\Lambda$ by increasing/decreasing 
the entry at $x\in \Lambda$ by $2\pi$.
Thus the corresponding spectral gap of this dynamics is the smallest constant $\gamma$ such that, for all
functions $F: (2\pi\Z)^\Lambda \to \R$ with finite variance,
\begin{equation}
  \var_\mu(F) \leq \frac{1}{\gamma} \frac{1}{2(2\pi)^2} \sum_{x\in\Lambda} 
 \bbE_\mu  \Big( (F(\sigma^{x+})-F(\sigma))^2  + (F(\sigma^{x-})-F(\sigma))^2 \Big).
\end{equation}

The following theorem is related to Theorem \ref{thm:gap-sg}.
It shows that the spectral gap of the Discrete Gaussian model scales like the one of the GFF.

\begin{theorem} \label{thm:gap-dg}
Let $\gamma_N(\beta)$ be the spectral gap of the hierarchical Discrete Gaussian model
on $\Lambda_N$ in dimension $d=2$ (as defined above).
For $\beta>0$ sufficiently small and $\epsilon = \beta L^{-2N}$,
there are $\kappa \in (0,1)$ and $c>0$ such that
\begin{equation}
   cL^{-2N} \leq \gamma_N(\beta) \leq
  L^{-2N} \p{ 1- O(\kappa^N)}
\end{equation}
provided that $N$ is sufficiently large.
\end{theorem}

\section{Induction on renormalised Brascamp--Lieb inequalities}
\label{sec:recursion}

The Brascamp--Lieb inequality is a generalisation of the spectral gap inequality.
We here say that a measure $\mu$ on a finite-dimensional vector space $X$ with inner product $(\cdot,\cdot)$
satisfies a Brascamp--Lieb inequality with quadratic form
$D :X \to X$ if for all smooth functions $F$,
\begin{equation}
\label{eq: Brascamp--Lieb inequality}
  \var_\mu(F) \leq \bbE_\mu(\nabla F, D\nabla F).
\end{equation}
In particular, if the quadratic form satisfies $D \leq \id / \lambda$ for some $\lambda>0$,
then $\mu$ satisfies a spectral gap inequality with constant $\lambda$. 
In this section, we construct inductive bounds on Brascamp--Lieb inequalities
between renormalised versions of a spin system.
From these we deduce in particular an induction on the spectral gap.
In the remainder of this paper, we will verify the generic assumptions made in this section
in the specific cases of the hierarchical $|\varphi|^4$ and the Sine-Gordon models.

\subsection{Hierarchical decomposition}

While the results of this section are somewhat more general,
in the remainder of this paper we will apply them to hierarchical models.
We therefore recall their structure which can be helpful to keep in mind throughout this section.
From Section~\ref{sec:hierGFF}, first recall
the orthogonal projections $P_1, \dots, P_N, Q_N$ whose ranges span $\R^\Lambda$,
and the hierarchical Laplacian $\Delta_H$ (see \eqref{e:DeltaH}).
By spectral calculus, for any $m^2 > 0$, its Green function can be written as
\begin{equation} \label{e:DeltaHdecompP}
  (-\Delta_H + m^2 )^{-1}
  = \sum_{j=1}^N
  (1+m^2L^{2(j-1)})^{-1}
  L^{2(j-1)} P_j
  + m^{-2} Q_N.
\end{equation}
Using the definition $P_j = Q_{j-1}-Q_j$ to express the right-hand side of the last equation in terms of the
block averaging operators $Q_j$,
we can alternatively write
\begin{equation} 
\label{e:DeltaHdecompQ}
  (-\Delta_H + m^2 )^{-1}
  =\sum_{j=0}^N C_j
  \quad \text{with}  \quad C_j = \lambda_j Q_j,
\end{equation}
where
\begin{gather}
\label{eq: tilde sigma}
  \lambda_0(m^2) = \frac{1}{1+m^2},
  \qquad
  \lambda_N(m^2)
  = \frac{1}{m^2(1+m^2L^{2(N-1)})},
  \\
  \lambda_j(m^2) = L^{2j} \frac{(1-L^{-2})}{(1+m^2L^{2j})(1+m^2L^{2(j-1)})}
  \quad (0<j<N).
\end{gather}
The above spin coupling matrices generalise directly
to the $O(n)$-invariant vector-valued case,
in which all operators act separately on each component,
and we use the same notation in this case.
Thus the Laplacian and the covariances act on the space $X_0 = \R^{n\Lambda}$.

The covariances $C_j$ are degenerate and 
it is convenient to introduce the subspaces of $X_0=\R^{n\Lambda}$ on which they are supported.
Thus define $X_j$ to be the image of $C_j$, i.e., 
\begin{equation} \label{e:Xjdef}
  X_j = \{ \varphi \in \R^{n\Lambda}: \text{$\varphi|_B$ is constant for every $B \in \cB_j$} \},
\end{equation}
and, for $S \subset \Lambda$,
\begin{equation}
  X_j(S) = \{ \varphi \in X_j: \varphi_x = 0 \text{ for } x \not\in S \}.
\end{equation}
Then the Gaussian field $\zeta = \{\zeta_x\}_{x\in\Lambda}$ with values in $X_j$ and covariance $C_j$ can be realised as
\begin{equation}
\forall x \in B, \qquad \zeta_x = \zeta_B,
\end{equation}
where  $\{ \zeta_B \}_{B \in \cB_j}$ are independent Gaussian variables 
in $\R^n$ with variance 
$ \frac{\lambda_j}{|B_{j}(x)|}= L^{-dj} \lambda_j$.

In general, one can identify $\varphi \in X_j$ with $\{\varphi_B\}_{B\in\cB_j}$.
In the following, we are going to consider functions defined only on the subspaces $X_j$.
Let $F$ be such a function of class $C^2$ written as 
\begin{equation}
\{ \varphi_B \}_{B \in \cB_j} \in \R^{n |\cB_j|}
\mapsto F \big( \{ \varphi_B \} \big).
\end{equation}
Then $F$ can be extended as a smooth function on the whole of $\R^{n\Lambda}$ by setting, for example,
\begin{equation}
  \label{eq: moyenne F}
  F(\varphi) = F\pB{\Big\{\frac{1}{|B|}\sum_{x\in B}\varphi_x\Big\} }.
\end{equation}
For such $F$,
we will consider the gradient and the Hessian of $F$ only in the directions spanned by $\b 1_B$ 
so that we set 
\begin{equation} \label{e:grad-He-Q}
  \forall \varphi \in X_j, \qquad 
  \nabla_{X_j} F (\varphi)  = Q_j\nabla F(\varphi ),
  \quad
  \He_{X_j} F (\varphi) = Q_j\He F(\varphi)Q_j .
\end{equation}
As the gradient and the Hessian are projected only in the directions spanned by $\b 1_B$,
their restrictions on $X_j$ are independent of the way $F$ has been extended in $\R^{n\Lambda}$.

\subsection{Renormalised measure}
\label{sec:renorm-measure}

Let $X_0=\R^{n\Lambda}$ with the standard inner product $(\cdot,\cdot)$.
From now on,
we consider a Gaussian measure on $X_0$ whose covariance $C_{\geq 0}$ has a decomposition $C_{\geq 0}=C_0 + \cdots + C_N$,
with the $C_i$ symmetric and positive semi-definite.
We then consider the class of probability measures $\mu$ with expectation
\begin{equation}
 \bbE_\mu(F) = \frac{\E_{C_{\geq 0}} (e^{-V_0}F)}{\E_{C_{\geq 0}} (e^{-V_0})},
\end{equation}
for some potential $V_0$.
In particular, the models introduced in Section~\ref{sec:intro} are in this class, with
\begin{equation}
  V_0(\varphi) = \sum_{x \in \Lambda} V(\varphi_x) \quad \text{for $\varphi \in X_0 = \R^{n\Lambda}$,}
\end{equation}
and the decomposition \eqref{e:DeltaHdecompQ}.
Given such a decomposition $C_0+ \cdots+C_N$ and the potential $V_0$, we define the renormalised potentials $V_j$ inductively by
\begin{equation} 
\label{e:sg-V+}
e^{-V_{j+1}(\varphi)} = \E_{C_{j}}(e^{-V_j(\varphi+\zeta)}),
\end{equation}
where the expectation applies to $\zeta$.
(This definition includes $j=N$, but throughout this section we will only use $j<N$.)
The associated renormalised measure $\mu_j$ is then defined by the expectation
\begin{equation}
\label{eq: renormalisation du potentiel}
\E_{\mu_j}(F) = \frac{\E_{C_{\geq j}}(e^{-V_{j}} F)}{\E_{C_{\geq j}}(e^{-V_j})}, \qquad C_{\geq j} = C_{j}+\cdots+C_N
  .
\end{equation}
As is the case for the hierarchical decomposition,
the covariances $C_j$ are permitted to be degenerate and
we denote by $X_j$ the subspaces of $X_0$ on which they are supported,
i.e., $X_j$ is the image of $C_j$ (see \eqref{e:Xjdef} for the hierarchical decomposition).

\subsection{One step of renormalisation}

For the remainder of the section, we fix a scale $j \in \{0,1,\dots, N\}$, 
and consider a single renormalisation group step
from scale $j$ to scale $j+1$ when $j<N$, and a final estimate when $j=N$.
To simplify the notation, we usually omit the scale index $j$ and write $+$ in place of $j+1$.
In particular, we write $C=C_j$, $V=V_j$, $\mu = \mu_{j}$, $\mu_{+} = \mu_{j+1}$, and so on.
Let $X = X_j \subseteq X_0$ be the image of $C$ 
and denote by $Q$ the orthogonal projection from $X_0$ onto $X$.
We need the following assumptions.

For $j<N$, in the assumptions below, $D_+=D_{j+1}$ is the matrix associated to a quadratic form
for a Brascamp--Lieb inequality for the measure $\mu_+$ (see \eqref{e:ass-BL}),
and we set $D_{N+1}=0$.
Throughout the paper, inequalities between operators and matrices are interpreted
in the sense of quadratic forms.

\bigskip
\noindent
{\bf A1. Non-convexity of potential.}
There is a constant $\gep = \gep_j < 1$ such that uniformly in $\gp \in X$, 
\begin{equation} \label{e:ass-V}
  E(\varphi) := C^{1/2} (\He_X V(\varphi)) C^{1/2}
  \geq -\gep Q.
\end{equation}
\noindent
{\bf A2. Coupling of scales.}
The images of $C$ and $C_+$ contain all directions on which $D_+$ is nontrivial,
more precisely
\begin{equation} \label{e:ass-coupling}
  D_+ = D_+Q  = D_+Q_+.
\end{equation}
\noindent
{\bf A3. Symmetry.}
For all $\varphi \in X$,
\begin{equation} \label{e:ass-commute}
  [E(\varphi),C]=[E(\varphi),D_+]=[C,D_+]=[C,Q_+]=0,
\end{equation}
where $[A,B] = AB-BA$ denotes the commutator.
\smallskip

The most significant assumption is \eqref{e:ass-V},
which will be seen to ensure that the fluctuation field measure given the block spin field is uniformly strictly convex.
The more technical assumptions \eqref{e:ass-coupling} and \eqref{e:ass-commute}
are very convenient (and obvious in the hierarchical setting \eqref{e:DeltaHdecompQ}) but seem less fundamental.
We use \eqref{e:ass-V} in Lemma~\ref{lem:BL-mugp} and \eqref{e:BL-A2-ass1},
\eqref{e:ass-coupling} in \eqref{e:BL-CS},
and
\eqref{e:ass-commute} in \eqref{e:BL-QVC}.

Under the above assumptions, we relate the Brascamp--Lieb inequality for $\mu_+$
to that for $\mu$.

\begin{theorem}
\label{thm:recurrence-BL}
Fix $j< N$, and assume (A1)--(A3) and that $\mu_+$ satisfies the Brascamp--Lieb inequality
\begin{equation} \label{e:ass-BL}
  \var_{\mu_+}(F)
  \leq \E_{\mu_+}(\nabla F(\varphi), D_+ \nabla F(\varphi)).
\end{equation}
Then $\mu$ satisfies a Brascamp--Lieb inequality \eqref{eq: Brascamp--Lieb inequality} with
\begin{equation}
\label{e:recurrence-BL}
  D
  \leq
  \frac{C}{1-\gep} + \frac{D_+}{(1-\gep)^2}
  \, .
\end{equation}
For $j=N$, assume only that (A1) holds.
Then $\mu$ satisfies a Brascamp--Lieb inequality \eqref{eq: Brascamp--Lieb inequality} with
\begin{equation}
\label{e:recurrence-BL-N}
  D
  \leq
  \frac{C}{1-\gep}
  \, .
\end{equation}
\end{theorem}

Iterating this theorem starting from $j=N$
gives the Brascamp--Lieb inequality for the original measure $\mu_{0}$ as follows.
In particular, the spectral gap of $\mu_0$ is bounded by the inverse of the largest eigenvalue of the matrix $D_0$.

\begin{corollary}
\label{cor: sum D0}
Assume that, for $j=0,\dots, N$, the sequence of renormalised measures $(\mu_{j})$ satisfies
Assumptions~(A1)-(A3) where $\epsilon =  \epsilon_j$.
Then $\mu_0$ satisfies a Brascamp--Lieb inequality with
  \begin{equation} \label{e:BL-rec-iterated}
    D_0
    \leq
    \sum_{k=0}^{N} \delta_k C_k,
    \qquad
    \delta_k = \frac{1}{1-\epsilon_k} \prod_{l=0}^{k-1} \frac{1}{(1-\epsilon_{l})^2}
    \leq \exp\pa{2\sum_{l=0}^{k} \epsilon_{l} + O(\epsilon_l^2)} .
  \end{equation}
\end{corollary}

\begin{proof}
By backward induction starting from $j=N$, we will prove that the renormalised measures $\mu_j$
satisfy the Brascamp--Lieb inequality
\begin{equation} \label{e:BL-induction}
  \var_{\mu_j}(F) \leq \E_{\mu_j}(\nabla F(\varphi), D_j\nabla F(\varphi)),
  \qquad
  \text{with }
  D_j \leq \sum_{k=j}^N \delta_{j,k} C_k
\end{equation}
and
\begin{equation}
  \delta_{j,k}= \frac{1}{1-\epsilon_k} \prod_{l=j}^{k-1} \frac{1}{(1-\epsilon_l)^2}
  .
\end{equation}
The claim \eqref{e:BL-rec-iterated} is then the case $j=0$.
To start the induction, we apply \eqref{e:recurrence-BL-N}
which gives \eqref{e:BL-induction} for $j=N$.
To advance the induction,
suppose $0 \leq j < N$ is such that the inductive assumption \eqref{e:BL-induction}
holds with $j$ replaced by $j+1$.
This means that \eqref{e:ass-BL} holds for $j$ and
Assumptions~(A1)--(A3) also hold by assumption of the corollary.
Theorem~\ref{thm:recurrence-BL} and the inductive assumption imply
that $\mu_j$ satisfies the Brascamp--Lieb
inequality with
\begin{equation}
  D_j
  \leq \frac{C_j}{1-\epsilon_j} + \frac{D_{j+1}}{(1-\epsilon_j)^2}
  \leq \frac{C_j}{1-\epsilon_j} +  \sum_{k=j+1}^{N} \frac{\delta_{j+1,k}}{(1-\epsilon_j)^2} C_k
  = \sum_{k=j}^{N} \delta_{j,k} C_k.
\end{equation}
This advances the inductive assumption,
i.e., \eqref{e:BL-induction} holds for $j$.
\end{proof}

\begin{corollary}
\label{cor: SG}
Under the assumptions of the previous corollary, the measure $\mu_0$ satisfies a spectral
gap inequality with inverse spectral gap less than the largest eigenvalue of the matrix $D_0$.
\end{corollary}

\begin{proof}
  The claim is immediate from the definitions of the Brascamp--Lieb and the spectral gap inequalities.
  Indeed, if $1/\lambda$ is the largest eigenvalue of $D_0$ then
  \begin{equation}
    \var_{\mu_0}(F) \leq \E_{\mu_0}(\nabla F, D_0 \nabla F) \leq \frac{1}{\lambda} \E_{\mu_0}(\nabla F, \nabla F),
  \end{equation}
  as claimed.
\end{proof}

In Sections~\ref{sec:phi4}--\ref{sec:sg},
Assumptions (A1)--(A3) will be checked for
the different hierarchical models in order to derive the scaling of the spectral gap 
from the previous corollary.

\begin{remark}
More generally,
in the assumption $D_+ = D_+(\varphi)$ and $\epsilon=\epsilon(\varphi)$ could depend on $\varphi \in X$,
with $\epsilon$ uniformly bounded by $1$.
The conclusion \eqref{e:recurrence-BL} is then replaced by
\begin{equation}
\label{e:recurrence-BL-varphi}
  D(\varphi+\zeta)
  \leq
  \frac{C}{1-\gep(\varphi+\zeta)} + \frac{D_+(\varphi)}{(1-\gep(\varphi+\zeta))^2}
  \, .
\end{equation}
However, this strengthened inequality may be difficult to use. To improve the readability,
we therefore do not carry the additional arguments for $D_+$ and $\epsilon$ through the proof.
\end{remark}

\subsection{Proof of Theorem~\ref{thm:recurrence-BL}}

We write the renormalised field at scale $j$ as $\zeta+\varphi$ where $\varphi\in X_+$
is the \emph{block spin field} at the next scale $j+1$
and $\zeta \in X$ is the \emph{fluctuation field} at scale $j$.
More precisely, recall that
\begin{equation}
  \E_{\mu}(F)
  = \frac{\E_{C_{\geq}}(e^{-V} F)}{\E_{C_{\geq}}(e^{-V})}
  = \frac{{\E_{C_>} \, \E_{C}} (e^{-V(\varphi+\zeta)} F(\varphi+\zeta))}{\E_{C_>} \, \E_C(e^{-V(\varphi+\zeta)})},
\end{equation}
where $C = C_j$ and $\zeta$ denotes the corresponding random field,
where $C_>$ stands for the covariance $C_{j+1} + C_{j+2} + \dots C_N$
and $\varphi$ denotes the corresponding random field,
where $C_{\geq} = C + C_>$,
and where $\E_C$ denotes the expectation of a Gaussian measure with covariance $C$.

Define the expectation conditioned on the block spin field $\varphi$ in $X_+$ by
\begin{equation} \label{e:condmeasdef}
  \E_{\mu_\varphi}(F) =  \E_{\mu}(F|\varphi)
  = \frac{\E_{C}(e^{-V(\varphi+\cdot)} F)}{\E_{C}(e^{-V(\varphi+\cdot)})}
  = \frac{\E_{C}(e^{-V(\varphi+\cdot)} F)}{e^{-V_+(\varphi)}}.
\end{equation}
where we will often use the notation $\E_{\mu_\varphi}$
for the conditional measure $\E_{\mu}(\cdot |\varphi)$ to make the notation more concise.
Then, using \eqref{eq: renormalisation du potentiel},
\begin{equation}
\bbE_{\mu}(F)
= \frac{1}{Z_{j+1}} \E_{C_>}
\pB{  e^{- V_{+} ( \varphi)} \; \bbE_{\mu} ( F | \varphi ) }
= \bbE_{\mu_+} \pB{ \bbE_{\mu} ( F | \varphi ) },
\end{equation}
where $Z_{j+1}$ is a normalising constant.

To prove Theorem~\ref{thm:recurrence-BL},
we write using the conditional expectation, 
\begin{equation}
\label{eq: decomposition variance}
  \E_{\mu} (F^2) - \E_{\mu}(F)^2 = 
  \E_{\mu_+} \pB{ \E_{\mu} \p{ F(\gp+\gz)^2 | \gp  } }
  - \E_{\mu_{+}} \pB{ \E_{\mu} \p{ F (\gp+\gz) | \gp  } }^2
  = \bbA_1 + \bbA_2,
\end{equation}
with 
\begin{align}
  \label{eq: identity variance}
  \bbA_1 &= 
           \E_{\mu_{+}} \pB{  \E_\mu \p{ F \p{\gp+\gz }^2 | \gp  }
           -  \E_\mu \p{ F \p{\gp+ \gz} | \gp  }^2 }
           ,
  \\
  \bbA_2 &=
           \E_{\mu_{+}} \pB{ \E_\mu \p{ F \p{  \gp+\gz  } | \gp  }^2 }
           - \E_{\mu_{+}} \pB{ \E_\mu \p{ F \p{  \gp+\gz } | \gp } }^2
.
\end{align}
In the remainder of this section, we will bound each term separately 
thanks to the following lemmas.
\begin{lemma}
\label{lem:A1}
Assume (A1). Then for any function $F$  with gradient in $L^2(\mu)$, one has
 \begin{align}
\label{eq: 1st term hier 0}
\bbA_1
\leq \E_\mu  \left(\nabla F(\varphi) \frac{C}{1 -\gep} \nabla F(\varphi) \right).
\end{align}
 \end{lemma}
 
\begin{lemma}
\label{lem:A2}
Assume (A1)--(A3) and that $\mu_+$ satisfies the Brascamp--Lieb inequality \eqref{e:ass-BL}. Then
for any function $F$ with gradient in $L^2(\mu)$, one has 
\begin{eqnarray}
\label{eq: 2nd term hier}
  \bbA_2 \leq \E_\mu\left( \nabla F(\varphi) \frac{D_+}{(1 -  \gep)^2} \nabla F(\varphi) \right) .
\end{eqnarray}
\end{lemma}

\begin{proof}[Proof of Theorem~\ref{thm:recurrence-BL}]
  For $j<N$, the proof is immediate by combining the decomposition \eqref{eq: decomposition variance}
  and the previous two lemmas. For $j=N$, the claim follows directly from Lemma~\ref{lem:A1} only.
\end{proof}

\subsubsection{Proof of Lemma \ref{lem:A1}}

From now on, we freeze the block spin field $\gp \in X_+$.
Then the conditional measure $\mu_\varphi = \mu( \,\cdot\, |\varphi)$
is a probability measure on the space $X$, the image of $C$
(see \eqref{e:Xjdef} in the hierarchical case).
As a subspace of the Euclidean vector space $X_0$,
the space $X$ has an induced inner product which we also denote by $(\cdot,\cdot)$,
and an induced surface measure, which is equivalent to the Lebesgue measure of the dimension of $X$.
The measure $\mu_\varphi$ has density proportional to $e^{-H_\gp(\zeta)}$ with respect to this measure given by
\begin{equation}
\label{eq: potential level j}
H_\varphi(\gz) = \frac{1}{2} (\gz  , C^{-1} \gz) + V(\gp+\gz).
\end{equation}
(By definition of the subspace $X$ we can regard $C$ as an invertible symmetric operator $X \to X$.)
For a function $F: X_0 \to \R$ and $\varphi \in X_0$,
the function $F_\varphi: X \to \R$ is defined by $F_\varphi(\zeta) =  F(\varphi+\zeta)$.

\begin{lemma} \label{lem:BL-mugp}
Assume (A1).
Then for all $\varphi \in X_+$,
the conditional measure $\mu_\varphi$ satisfies the Brascamp--Lieb inequality
\begin{equation}
\label{e:BL-term1}
\E_{\mu_\varphi} (F_\gp(\gz)^2) - \E_{\mu_\varphi} (  F_\gp(\gz))^2 
\leq 
\E_{\mu_\varphi} \left( \big(\nabla_X F_\gp(\gz), \frac{C}{1 -\gep}\nabla_X F_\gp(\gz)
\big) \right)  .
\end{equation}
\end{lemma}

\begin{proof}
As a consequence of Assumption~\eqref{e:ass-V} and of the definition of the space $X$,
the Hamiltonian $H_\varphi$ associated with $\mu_\varphi$ is strictly convex on $X$, with
\begin{align*}
  \He_{X} H_\varphi
  &= 
  C^{-1} + \He_X V_\varphi
  \\
  &=
  C^{-1/2}(\id + C^{1/2} \He V_\varphi C^{1/2})C^{-1/2}
  \geq (1-\gep) C^{-1},
\end{align*}
where we used that $C$ is invertible on $X$ and that $QC = CQ= C$.
The Brascamp--Lieb inequality \eqref{e:BL0} implies the inequality.
\end{proof}

\begin{proof}[Proof of Lemma~\ref{lem:A1}]
The term $\bbA_1$ is a variance under the conditional measure $\mu_\varphi$.
By Lemma~\ref{lem:BL-mugp}, the measure satisfies the Brascamp--Lieb inequality \eqref{e:BL-term1}.
Therefore
\begin{align}
\label{eq: 1st term hier}
\bbA_1
&=\E_{\mu_{+}} \Big(  \E_{\mu_\gp}( F_\gp(\gz)^2) -  \mu_\gp( F_\gp( \gz))^2 \Big)
\nnb
&\leq \E_{\mu_{+}} \Big( \mu_\varphi\Big(\nabla_X F_\gp(\gz)  \frac{C}{1 -\gep} \nabla_X F_\gp(\gz) \Big) \Big)
= \E_\mu  \left(\nabla F(\varphi) \frac{C}{1 -\gep} \nabla F(\varphi) \right).
\end{align}
In the last equality we used that $CQ=C$ by definition of $Q$
as the orthogonal  projection onto the image of $C$ so that 
 $\nabla_X$ can be replaced by $\nabla$.
\end{proof}

\subsubsection{Proof of Lemma \ref{lem:A2}}

The second term $\bbA_2$ in \eqref{eq: identity variance} is a variance under $\mu_{+}$:
\begin{equation}
  \bbA_2 = \E_{\mu_{+}} \Big( \tilde F(\gp)^2 \Big) - \E_{\mu_{+}} \Big(\tilde F(\gp)\Big)^2,
  \quad
  \tilde F(\gp)
  = \E_{\mu_\gp}( F_\varphi(\gz)).
\end{equation}
Using Assumption~\eqref{e:ass-BL} that the measure $\mu_{+}$ satisfies a Brascamp--Lieb inequality, we have
\begin{equation} \label{e:A2bd}
\bbA_2 
\leq
\E_{\mu_{+}} \Big( \| D_+^{1/2} \nabla \tilde F(\gp)\|_2^2 \Big)
=
\E_{\mu_{+}} \Big( \| D_+^{1/2} \nabla_{X_+} \E_{\mu_\gp}( F(\gp + \gz))\|_2^2 \Big)
\, ,
\end{equation}
where $\nabla_{X_+}$ applies to the variable $\varphi$
and $\|f\|_2^2 = \sum_{x\in\Lambda} |f_x|^2$.

We first state a technical lemma.
\begin{lemma}
  Assume (A3). For $\dot\varphi \in X_+$,
  \begin{equation}
    \label{e:derivative2}
    (\dot\varphi, \nabla_{X_+} \tilde F(\gp))
    =
    (\dot\varphi, \nabla_{X_+} \E_{\mu_\gp} ( F (  \gp + \gz ) ))
    =  \cov_{\mu_\gp}( F(\gp+\gz) ,  \, (\dot\varphi, C^{-1}\gz)).
  \end{equation}
\end{lemma}

\begin{proof}
The derivative applies only on the block spin field $\gp$.
We write $\nabla_\varphi$ for $\nabla_{X_+}$ with respect to the variable $\varphi$
and $\nabla_\zeta$ for $\nabla_{X}$ with respect to the variable $\zeta$.
Using the notation \eqref{eq: potential level j}, 
\begin{align}
\label{e:derivative}
  (\dot\varphi, \nabla_{\gp} \E_{\mu_\gp} \left( F \big(  \gp + \gz  \big)    \right) )
  & =  
    \E_{\mu_\gp} \left( (\dot\varphi, \nabla_{\gp}   F \big(  \gp + \gz  \big) ) \right)
    -  \cov_{\mu_\gp} \left( F \big( \gp + \gz \big) \, ,  \, (\dot\varphi,\nabla_{\gp}  H_\gp(\gz))  \right)
    \nonumber
 \\
  & =  
    \E_{\mu_\gp} \left( (\dot\varphi,\nabla_{\gz}   F \big(  \gp + \gz  \big) ) \right)
    - \cov_{\mu_\gp} \left( F \big( \gp + \gz \big) \, ,  \, (\dot\varphi,\nabla_\zeta  V \big(  \gp + \gz  \big))  \right),
\end{align}
where in the last term we used that, since $\dot\gp \in X_+$,
\begin{equation}
  (\dot\varphi,\nabla_{\gp}  F) = (\dot\varphi, \nabla_{\gz} F),
  \qquad
  (\dot\varphi,\nabla_{\gp}  H_\varphi) = (\dot\varphi, \nabla_{\gz}  V).
\end{equation}
By integration by parts, we get also that
\begin{equation}
  \E_{\mu_\gp} \left( \nabla_\gz  F \big(  \gp + \gz  \big) \right)
  =
  \E_{\mu_\gp} \left(  F(\gp + \gz) \nabla_\zeta H_{\gp}(\gz )  \right)
  .
\end{equation}
Using this relation and \eqref{e:ass-commute},  we get that for any $\zeta \in X$,
\begin{equation}
  (\dot\varphi,\nabla_{\zeta} H_{ \gp} (\gz))
  = (\dot\varphi, \nabla_{\zeta} \frac{1}{2} (\gz  , C^{-1} \gz)) + (\dot\varphi, \nabla_{\zeta} V(\gp+ \gz))
  = (\dot\varphi, C^{-1} \gz) + (\dot\varphi, \nabla_\zeta V(\gp+\gz)),
\end{equation}
and therefore
\begin{equation}
  \E_{\mu_\gp} \left( (\dot\varphi, \nabla_\gz  F \big(  \gp + \gz  \big))\right)
 =
  \E_{\mu_\gp} \left(  F(\gp + \gz) (\dot\varphi, C^{-1} \zeta)  \right)
  +
  \E_{\mu_\gp} \left(  F(\gp + \gz) (\dot\varphi, \nabla_\zeta V(\gp+\gz))  \right)
.
\end{equation}
The last equality applied to $F=1$ implies that (as an identity between elements of $X_+$)
\begin{equation}
  \E_{\mu_\gp} \left( (\dot\varphi, \nabla_{\gz} V(\gp+\gz )) \right)
  = - \E_{\mu_\gp}((\dot\varphi,C^{-1}\gz)) .
\end{equation}
Thus \eqref{e:derivative} becomes
\begin{equation}
\label{eq: derivative 2 hier bis}
  (\dot\varphi, \nabla_{\gp} \E_{\mu_\gp} \left( F \big(  \gp + \gz  \big)    \right)  )
  = \cov_{\mu_\gp} \left( F \big( \gp + \gz \big) \,,  (\dot\varphi, C^{-1} \zeta)  \right),
\end{equation}
as claimed.
\end{proof}

\begin{lemma} \label{lem:secondterm}
Assume (A1)--(A3).
Then for $\gp$ in $X_+$,
\begin{equation}
  \|D_+^{1/2} \nabla_{X_+} \E_{\mu_\gp} \left( F(\gp + \gz)    \right)\|_2^2
  \leq
  \E_{\mu_\gp} \left( \|\frac{D_+^{1/2}}{1-\gep}\nabla_{X_+} F(  \gp + \gz )\|_2^2     \right)
  =
  \E_{\mu_\gp} \left( \|\frac{D_+^{1/2}}{1-\gep}\nabla F(  \gp + \gz )\|_2^2     \right)
  .
\end{equation}
\end{lemma}

Applying the expectation $\E_{\mu_+}( \cdot)$ on both sides and substituting the result into \eqref{e:A2bd},
this completes Lemma \ref{lem:A2}.

\begin{proof}[Proof of Lemma~\ref{lem:secondterm}]
The block spin field $\gp \in X_+$ is fixed and in the proof we study the measure $\mu_\varphi$ on the subspace $X$.  
We define $L_\varphi$ to be the self-adjoint generator of the Glauber dynamics
for the conditional measure $\mu_\varphi$ on $X$, i.e., 
\begin{equation}
\label{e:langevin}
L_\varphi F (\gz) = \Delta_{X} F (\gz) + (\nabla_{X} H_\varphi(\gz),\nabla_X F (\gz));
\end{equation}
see also Appendix~\ref{app:convex}.
Moreover, we define the Witten Laplacian $\cL_\varphi$ on $L^2(\mu_\varphi) \otimes X$ by
\begin{equation}
\cL_\varphi = L_\varphi \otimes \id_X + \He_{X} H_\varphi \, .
\end{equation}
Using the Helffer-Sj\"ostrand representation (Theorem~\ref{thm:HS}),
one can rewrite the correlations \eqref{e:derivative2} under the conditional measure in terms
of the operator $\cL_\varphi$ as
\begin{align} \label{e:CnablaE}
  (\dot\gp, \nabla_{X_+} \E_{\mu_\varphi}( F (  \gp + \gz )))
  &=  \cov_{\mu_\gp}( F(\gp+\gz) ,  \, (C^{-1}\gz, \dot\gp))
    \nnb
  &= \E_{\mu_\varphi}(\nabla_{X} (C^{-1} \gz,\dot\gp) , \cL_\varphi^{-1}  \, \nabla_{X}  F ( \gp + \gz))
  \nnb
  &= (C^{-1} \dot\gp, \E_{\mu_\varphi}(\cL_\varphi^{-1}  \, \nabla_{X}  F ( \gp + \gz)))
  \nnb
  &= (\dot\gp, \E_{\mu_\varphi}(C^{-1}\cL_\varphi^{-1}  \, \nabla_{X}  F ( \gp + \gz))).
\end{align}
This is an identity  in $X_+$ which can be rewritten by using the projection $Q_+$ as
\begin{equation}
\nabla_{X_+} \E_{\mu_\varphi}( F (  \gp + \gz )) = 
\E_{\mu_\varphi}(Q_+ C^{-1}\cL_\varphi^{-1}  \, \nabla_{X}  F ( \gp + \gz)).
\end{equation}
Composing by $D_+^{1/2}$ and using that $D_+ = D_+ Q_+$ by \eqref{e:ass-coupling}, we deduce that 
\begin{equation}
\label{e:CnablaE 2}
D_+^{1/2} \nabla_{X_+} \E_{\mu_\varphi}( F (  \gp + \gz )) = 
\E_{\mu_\varphi}( M_\varphi  \, \nabla_{X}  F ( \gp + \gz)).
\end{equation}
where the operator $M_\varphi$ is defined as
\begin{equation}
  M_\varphi = D_+^{1/2}C^{-1} \cL_\varphi^{-1}.
\end{equation}
Since $D_+$ commutes with $C$ and with $\cL_\varphi C$ by \eqref{e:ass-commute},
the operator $M_\varphi$ acts on $L^2(\mu_\varphi) \otimes X$ and is self-adjoint.
From \eqref{e:CnablaE 2} and the Cauchy-Schwarz inequality,
we finally obtain
\begin{equation} \label{e:BL-CS}
  \|D_+^{1/2} \nabla_{X_+} \E_{\mu_\gp}(F(\gp+\gz))\|_2^2
  \leq \E_{\mu_\gp} \pB{ \|M_\varphi  \nabla_{X}  F(\gp+\gz)\|_2^2  },
\end{equation}
where $\|f\|_2^2 = (f,f)$ and $\nabla_{X_+}$ applies to $\varphi$ and $\nabla_X$ applies to $\zeta$.
In the following, we will show that the operator $M_\varphi$
obeys the following form inequality on $L^2(\mu_\varphi)\otimes X$:
\begin{equation}
  M_\varphi \leq (1-\gep)^{-1} D_+^{1/2},
\end{equation}
 which then concludes the proof of the lemma.
Recall that the operator $\cL_\gp$ is defined by
\begin{equation}
\cL_\gp
= L_\varphi \otimes \id_{X} + \He_X H_\varphi
= L_\varphi \otimes \id_{X} + \He_X V(\gp+\gz) + C^{-1}.
\end{equation}
Under Assumption~\eqref{e:ass-commute}, we can write
\begin{equation} \label{e:BL-QVC}
  (\He_X V)C
  = C^{1/2} (\He_X V) C^{1/2}.
\end{equation}
Using that $L_\varphi$ and $C$ are positive operators,
using Assumption~\eqref{e:ass-V},
it follows that
as operators on $L^2(\mu_\varphi) \otimes X$,
\begin{equation} \label{e:BL-A2-ass1}
  \cL_\varphi C
  = C^{1/2} \cL_\varphi C^{1/2}
  = L_\varphi \otimes C + \id_X + C^{1/2}(\He _X V (\gp+\gz) )C^{1/2}
  \geq (1 -\gep) Q 
  .
\end{equation}
Finally, using that $D_+=D_+Q$ by Assumption~\eqref{e:ass-coupling},
and using \eqref{e:ass-commute}, it follows that $M_\varphi$
satisfies the desired form bound
\begin{equation}
  M_\varphi \leq (1-\epsilon)^{-1}  D_+^{1/2}.
\end{equation}
This completes the proof.
\end{proof}

\section{Hierarchical $|\varphi|^4$ model}
\label{sec:phi4}

In this section, we apply Corollaries~\ref{cor: sum D0}--\ref{cor: SG} to the hierarchical $|\varphi|^4$ model.
Throughout this section, the dimension is fixed to be $d=4$.
Nevertheless, we sometimes write $d$ to emphasise that a factor $4$ arises from the dimension $d=4$
rather than from the exponent of $|\varphi|^4$.

\subsection{Renormalisation group flow}
\label{sec:phi4-rg}

\newcommand{\Vstep}{\hat V}
\newcommand{\Wstep}{\hat W}
\newcommand{\Kstep}{\hat K}

For $m^2>0$ (to be determined in Theorem~\ref{thm:phi4-rg} as a function of $g$ and $\nu$), 
we decompose
\begin{equation}
  (-\Delta_H+m^2)^{-1} = C_0 + \cdots + C_N,
\end{equation}
as in \eqref{e:DeltaHdecompQ}, and define the renormalised potential with respect to this decomposition
as in \eqref{e:sg-V+},
\begin{equation} \label{e:phi4-Vdef}
  e^{-V_{j+1}(\varphi)}
  = \E_{C_{j}} \pa{ e^{-V_j(\varphi+\zeta)} }.
\end{equation}
Note in particular that the sequence of renormalised potentials depends on the choice of $m^2$,
and that $C_j \leq \vartheta_j^2 L^{2j} Q_j$ where we define $\vartheta_j = 2^{-(j-j_m)_+}$.
As a consequence of the hierarchical structure, the renormalised potential can be written as
\begin{equation} \label{e:tildeV}
  V_j(\varphi) = \sum_{B\in \mathcal{B}_j} V_j(B,\varphi),
\end{equation}
where $V_j(B,\varphi)$ is a function of $\varphi$ that depends only on the restriction $\varphi|_B$ for any block $B \in \cB_j$.

We always restrict the domain of the functions $V_j(B)$ to the space $X_j(B) \cong \R^n$
of fields that are constant on $B$. 
Explicitly, for a block $B \in \mathcal{B}$,
denote by $i_B: \R^n \to \R^{nB}$ the linear map that sends $\varphi \in \R^n$ to the constant field $\varphi : B \to \R^n$
with $\varphi_x = \varphi$ at every $x \in B$.
Then $V_j(B) \circ i_B$ is a function of a single variable in $\R^n$ induced by $V_j(B)$.
In particular using \eqref{eq: moyenne F} one can view  $V_j(B)$ as a function in $\R^{nB}$, so that
for any ${\dot\varphi} \in X_j(B)$ taking the constant value ${\dot\varphi}_B \in \R^n$,
\begin{equation} \label{e:HeVB}
  {\dot\varphi} (\He V_j(B)) {\dot\varphi}
  = 
  {\dot \varphi_B} \He (V_j(B) \circ i_B)\dot\varphi_B.
\end{equation}
If there is a  constant $s >0$ such that 
\begin{equation} 
\label{e:HeVB borne inf}
\frac{1}{ |B| }  {\dot \varphi_B} \He (V_j(B) \circ i_B) \dot \varphi_B 
\geq - s  (\dot\varphi_B,\dot\varphi_B),
\end{equation}
then using that $(\dot\varphi,\dot\varphi) =|\dot\varphi_B|^2|B|$, we deduce 
\begin{equation} 
 {\dot\varphi} (\He V_j(B)) {\dot\varphi}
\geq - s (\dot\varphi ,\dot\varphi ).
\end{equation}
With the notation \eqref{e:grad-He-Q}, the inequalities \eqref{e:HeVB borne inf}
and $C_j \leq \vartheta_j^2 L^{2j} Q_j$, it follows that
\begin{equation} 
\label{e:HessiB}
C_j^{1/2}(\He_{X_j} V_j)C_j^{1/2}
\geq - s \vartheta_j^2 L^{2j}  Q_j .
\end{equation}
Thus, in the hierarchical model, Assumption~(A1) in \eqref{e:ass-V}
with $\epsilon_j = s \vartheta_j^2 L^{2j}$
follows from \eqref{e:HeVB borne inf}. In the rest of this section, we therefore
reduce to the study of the function $V_j(B) \circ i_B$ in $\R^n$.

The renormalisation group for the $|\varphi|^4$ model provides precise estimates on the
renormalised potential $V_j$ when the field $\varphi$ is not too large.
The following theorem about the renormalisation group flow is proved in \cite{rg-brief}.
Note that $V_j$ in \eqref{e:phi4-Vdef} is the full renormalised potential
(the logarithm of the density with respect to the Gaussian reference measure),
not its leading contribution as in \cite{rg-brief}.
We will denote the latter instead by $\Vstep_j$ as it plays a less central role in the arguments of this paper.
It is determined by the coupling constants $(g_j,\nu_j) \in \R^2$ through
\begin{equation} \label{e:VWstep}
  \Vstep_j(B,\varphi) = \sum_{x\in B} \pa{ \frac14 g_j|\varphi_x|^4 + \frac12\nu_j|\varphi_x|^2},
  \quad \Wstep_j(B,\varphi) = \sum_{x \in B} \pa{ \frac16 \alpha_j g_j^2 |\varphi_x|^6},
\end{equation}
where $\alpha_j = \alpha_j(m^2) = O(L^{2j}L^{-(j-j_m)_+})$ is an explicit ($j$-dependent) constant
and $j_m = \floor{\log_L m^{-1}}$ is the mass scale.
We stress the fact that if the field is constant on $B$ then  
\begin{equation} 
\label{e:VWstep 2}
  \Vstep_j(B) \circ i_B (\varphi) = |B| \pa{ \frac14 g_j|\varphi|^4 + \frac12\nu_j|\varphi|^2},
  \quad 
  \Wstep_j(B) \circ i_B (\varphi) = |B| \pa{ \frac16 \alpha_j g_j^2 |\varphi|^6},
\end{equation}
so that in the following we will often consider the effective potential normalised by the factor $1/ |B|$
(see also \eqref{e:HeVB borne inf}).

For the statement of the theorem, define
the \emph{fluctuation field scale} $\ell_j$ and the \emph{large field scale} $h_j$ by
\begin{equation}
  \label{eq: scales}
  \ell_j =  L^{-(d-2) j/2} = L^{-j},
  \qquad
  h_j = L^{-dj/4}g_j^{-1/4} = L^{-j} g_j^{-1/4}.
\end{equation}
Finally, we define $\cF_j$ by $F \in \cF_j$ if for any $B\in\cB_j$ there is a function
$\varphi \in \R^{n\Lambda} \mapsto F(B,\varphi)$ that (i) depends only on the average of $\varphi$ over the block $B$;
(ii) the function $F(B) \circ i_B$ is the same for any block $B$;
and (iii) the function $F(B)$ is invariant under rotations, i.e., $F(\varphi,B) = F(T\varphi,B)$ for any $T \in O(n)$
acting on $\varphi \in \R^{n\Lambda}$ by $(T\varphi)_x = T\varphi_x$; see \cite[Definition~5.1.5]{rg-brief}.

\begin{theorem} \label{thm:phi4-rg}
Let $L \geq L_0$.
For any $g>0$ small enough, there exists $\nu_c(g) = -C(n+2)g + O(g^2)$ such that for
$\nu > \nu_c(g)+cL^{-2N}$, there exists $m^2 > 0$,
a sequence of coupling constants $(g_j,\nu_j, u_j) \subset \R^3$,
and $\hat K_j \in \cF_j$ 
such that the following are true.
\begin{enumerate}
\item
  The full renormalised potential $V_j$ defined by \eqref{e:phi4-Vdef} satisfies:
  for all $\varphi$ that are constant on $B$,
\begin{equation} \label{e:phi4-rg-repr}
  e^{-V_j(B,\varphi)} = e^{-u_j|B|}(e^{-\Vstep_j(B,\varphi)}(1+\Wstep_j(B,\varphi)) + \Kstep_j(B,\varphi)).
\end{equation}
\item
  The sequence $(g_j,\nu_j)$ of coupling constants satisfies $(g_0,\nu_0)=(g,\nu-m^2)$, and
  \begin{equation} \label{e:phi4-rg-gnubd}
    g_{j+1} = g_j - \beta_j g_j^2 + O(2^{-(j-j_m)_+}g_j^3),
    \qquad
    0 \geq L^{2j}\nu_j = O(2^{-(j-j_m)_+}g_j),
  \end{equation}
  where $\beta_j = \beta_0^0(1+m^2L^{2j})^{-2}$ for an absolute constant $\beta_0^0>0$ and $j_m = \floor{\log_Lm^{-1}}$.
\item
  The functions $\hat K_j$ satisfy $\hat K_0=0$ and
  \begin{align} \label{e:phi4-rg-Tinfty}
    \sup_{\varphi \in \R^n} \max_{0\leq \alpha \leq 3} h_j^{\alpha} |\nabla^\alpha (\Kstep_j(B) \circ i_B)(\varphi)| &= O(2^{-(j-j_m)_+}g_j^{3/4}),
    \\
    \label{e:phi4-rg-T0}
    \max_{0\leq \alpha \leq 3} \ell_j^{\alpha} |\nabla^\alpha (\Kstep_j(B) \circ i_B)(0)| &= O(2^{-(j-j_m)_+}g_j^{3}),
  \end{align}
  where $\ell_j = L^{-j}$ and $h_j = L^{-j} g_j^{-1/4}$.
\item
  The relation between $t = \nu - \nu_c(g) >0$ and $m^2>0$ satisfies, as $t \downarrow 0$,
  \begin{equation} \label{e:phi4-rg-mepsilon}
    m^2 \sim C_g t(\log t^{-1})^{-(n+2)/(n+8)}.
  \end{equation}
\end{enumerate}  
\end{theorem}

In the above theorem and everywhere else, the error terms $O(\cdot)$ are uniform in the scale $j$.
The theorem is mainly proved and explained in \cite{rg-brief}.
For our application to the analysis of the spectral gap of the Glauber dynamics,
it is however more convenient to use a slightly different organisation than that used in \cite{rg-brief}.
It is here better to use the decomposition \eqref{e:DeltaHdecompQ}
instead of \eqref{e:DeltaHdecompP} (used in \cite{rg-brief}).
We translate between the conventions in \cite{rg-brief} and those used in the statement of Theorem~\ref{thm:phi4-rg} in Appendix~\ref{app:phi4-rg-pf}
and also give precise references there.

We remark that the normalising constants $u_j$ are unimportant for our purposes, and that
the recursion \eqref{e:phi4-rg-gnubd} implies that, as $m^2 \downarrow 0$,
\begin{equation} \label{e:gjm}
  g_j^{-1} = O(g_{j_m}^{-1}),
  \qquad
  g_{j_m}^{-1} \sim \beta_0^0 \log m^{-1};
\end{equation}
see \cite[Proposition~6.1.3]{rg-brief}.

\bigskip

A variant of the theorem implies the following asymptotic behaviour of the susceptibility as the critical point is approached.

\begin{corollary}
  Let $F= \sum_x\varphi_x^{1}$. Then for $t = \nu-\nu_c \geq c L^{-2 N}$,
  \begin{equation} \label{e:phi4-varF}
    \frac{ \var_\mu(F)}{|\Lambda_N|} = \frac{1}{m^2} \pa{1+o\pa{\frac{1}{L^{2N}m^2}}}
    \sim
    C_g \frac{1}{t}(-\log t)^{(n+2)/(n+8)},
  \end{equation}
  with $o(1)$ tending to $0$ as $L^{2N}m^2 \to \infty$,
  and $\var_\mu$ denotes the variance under the full $|\varphi|^4$ measure
  as in \eqref{eq: mu mesure}.
\end{corollary}

Indeed, the corollary is
\cite[Theorem~5.2.1 and (6.2.17)]{rg-brief},
noting that $\var_\mu(F)/|\Lambda_N|$ is the finite volume susceptibility studied there.
The corollary provides the upper bound in Theorem~\ref{thm:gap-phi4}
since, with $F$ as defined in the corollary,
\begin{equation}
  \frac{(\nabla F,\nabla F)}{|\Lambda_N|}  = 1,
\end{equation}
and $\gamma_N(g,\nu_c(g)) \leq \var_\mu(F)/\E_\mu(\nabla F,\nabla F)$ for any $F$
by definition of the spectral gap.

\subsection{Small field region}
\newcommand{\tj}{\vartheta_j}

The bounds of Theorem~\ref{thm:phi4-rg} are effective for small fields $|\varphi| \leq h_j$.
For such fields $\varphi$, the approximate effective potential
$\Vstep_j(\varphi)$ is a good approximation to $V_j(\varphi)$.
Indeed, then $e^{\Vstep_j(B,\varphi)} = e^{O(1)}$ and
\begin{align} 
\label{e:V-Vstep}
  V_j(B,\varphi)-\Vstep_j(B,\varphi)
  &=
  -\log (1+\Wstep_j(B,\varphi) + e^{\Vstep_j(B,\varphi)} \Kstep_j(B,\varphi)) + u_j |B|
  \nnb
  &= -\Wstep_j(B,\varphi)- e^{\Vstep_j(B,\varphi)} \Kstep_j(B,\varphi) + u_j |B|
  + O(\Wstep_j+e^{\Vstep_j} \Kstep_j)^2.
\end{align}
Recall the abbreviation $\tj = 2^{-(j-j_m)_+}$ where $j_m = \lfloor \log_L m^{-1} \rfloor$ is the mass scale.
By \eqref{e:phi4-rg-gnubd} and \eqref{e:phi4-rg-Tinfty} and the definition of $\Wstep$,
uniformly in $\varphi \in \R^n$ with $|\varphi| \leq h_j$,
\begin{align}
  \max_{0\leq \alpha \leq 3} h_j^{\alpha} |\nabla^\alpha (\Wstep_j(B) \circ i_B)(\varphi)| &= O(\tj g_j^{2/4}),
  \\
  \max_{0\leq \alpha \leq 3} h_j^{\alpha} |\nabla^\alpha (e^{\Vstep_j(B)}\Kstep_j(B) \circ i_B)(\varphi)| &= O(\tj g_j^{3/4}),
\end{align}
and the remainder satisfies an analogous estimate.
In particular, by \eqref{e:V-Vstep},
\begin{align} \label{e:V-VWstep}
  \He (V_j(B) \circ i_B)(\varphi)
  &= \He ((\Vstep_j-\Wstep_j)(B) \circ i_B)(\varphi) + O(\tj h_j^{-2}g_j^{3/4})\id_{n}
  \nnb
  &= \He ((\Vstep_j-\Wstep_j)(B) \circ i_B)(\varphi) + O(\tj L^{2j}g_j^{5/4}) \id_{n}
    ,
\end{align}
where $\id_{n}$ is the identity matrix acting on the single-spin space $\R^n$.
The first term on the right-hand side can be computed explicitly from \eqref{e:VWstep},
which implies that as quadratic forms,
\begin{align} \label{e:HetildeV}
  \frac{1}{|B|} \He (\Vstep_j(B) \circ i_B)(\varphi)
  &=
  ((g_j|\varphi|^2 +\nu_j)\id_n + 2g_j (\varphi^k\varphi^l)_{k,l})
  \geq
  \big(   g_j|\varphi|^2 + \nu_j  \big) \id_{n},
  \\
  \label{e:HetildeW}
  \frac{1}{|B|}|\He (\Wstep_j(B) \circ i_B)(\varphi)|
  &\leq
  5\alpha_j g_j^2(|\varphi|^4 \id_n + 2|\varphi|^2(\varphi^k\varphi^l)_{k,l})
  \leq (15\alpha_j g_j^2|\varphi|^4) \id_n,
\end{align}
where $|B| =L^{dj}$,
and where we used that the $n\times n$ matrix $(\varphi^k\varphi^l)_{k,l}$ has eigenvalues $0$ and $|\varphi|^2 \geq 0$.
Combining \eqref{e:V-VWstep} with \eqref{e:HetildeV}--\eqref{e:HetildeW}, we find that
\begin{equation} \label{e:phi4-small-HeV-bis}
  \frac{1}{|B|} \He (V_j(B) \circ i_B)(\varphi)
  \geq \pB{g_j|\varphi|^2 + \nu - 15\alpha_j g_j^2 |\varphi|^4 - O(\tj L^{-2j}g_j^{5/4})} \id_n.
\end{equation}
Using that $\alpha_j g_j|\varphi|^2 = O(g_j^{1/2})$ for $|\varphi| \leq h_j$ (since $\alpha_j = O(L^{2j})$),
in summary, we have obtained the following corollary of Theorem~\ref{thm:phi4-rg}.

\begin{corollary} \label{cor:phi4-small}
Suppose that $V_0$ satisfies the conditions of Theorem~\ref{thm:phi4-rg}.
Then for all scales $j\in\N$ and all $\varphi \in \R^n$ with $|\varphi|\leq h_j$,
the effective potential satisfies the quadratic form bounds
\begin{equation} \label{e:phi4-small-HeV}
  \frac{1}{|B|}\He (V_j(B) \circ i_B)(\varphi)
  \geq \pB{g_j|\varphi|^2(1-O(g_j^{1/2})) + \nu_j - O(\tj L^{-2j}g_j^{5/4})}\id_n,
\end{equation}
with $0 \leq -\nu_j = O(\tj L^{-2j}g_j)$, and furthermore
\begin{equation} \label{e:phi4-small-dV}
  \frac{1}{|B|} \nabla (V_j(B) \circ i_B)(\varphi)
  = g_j\varphi|\varphi|^2(1-O(g_j^{1/2})) + \nu_j\varphi +   O(\tj L^{-3j} g_j).
\end{equation}
\end{corollary}

\subsection{Large field region}

Using the small field estimates as input,
we are going to prove the following estimate for the large field region.

\begin{theorem} \label{thm:phi4-convex}
Assume the conditions of Theorem~\ref{thm:phi4-rg},
in particular that $g>0$ is sufficiently small and that $\nu > \nu_c(g) + cL^{-2N}$.
Then for all $j \in \N$ and all $B \in \cB_j$,
the effective potential satisfies
\begin{equation} \label{e:phi4-convex}
  L^{2j} \frac{1}{|B|} \He (V_j(B) \circ i_B) \geq \epsilon_j \id_n
  \quad
  \text{for all $\varphi\in\R^n$ with $|\varphi| \geq h_j$,}
\end{equation}
where the constants $\epsilon_j$ satisfy 
$\epsilon_{j+1} = \bar\epsilon_j - O(\vartheta_j^2\bar\epsilon_{j}^2)$ and $\epsilon_0= \frac15 g_0^{1/2}$ 
where $\bar\epsilon_j  = \epsilon_j \wedge \frac15 g_j^{1/2}$.
\end{theorem}

To prove Theorem~\ref{thm:gap-phi4},
we will only use the conclusion $\epsilon_j \geq 0$ from Theorem~\ref{thm:phi4-convex}.
However, in order to prove Theorem~\ref{thm:phi4-convex}, it is convenient that the $\epsilon_j$ do not become too small.
The elementary proof of the following estimate is given in Appendix~\ref{app:phi4-rg-pf}.

\begin{lemma}\label{lem:epsilon}
  The sequence $(\epsilon_j)$ defined 
  in Theorem~\ref{thm:phi4-convex}
  satisfies $\epsilon_j \geq c g_j$ for all $j\in\N$.
\end{lemma}

We will prove Theorem~\ref{thm:phi4-convex} by induction in $j$.
For $j=0$, the estimate \eqref{e:phi4-convex} 
can be checked directly from \eqref{e:HetildeV} and $\nu \geq \nu_c(g) = -O(g)$,
which imply that
\begin{equation}
  \frac{1}{|B|} \He (V_0(B) \circ i_B)
  \geq (g|\varphi|^2 + \nu)\id_n
  \geq g(|\varphi|^2 - O(1))\id_n
  \geq (g^{1/2} - O(g)) \id_n.
\end{equation}
From the inductive assumption and Corollary~\ref{cor:phi4-small}, we can get the following bounds.

\begin{lemma} 
\label{lem:ind-ass}
Assume that \eqref{e:phi4-convex} holds for some $j\in\N$ and that $\epsilon_j \leq \frac14 g_j^{1/2} - O(g_j)$. Then
\begin{align} \label{e:ind-ass1}
  L^{2(j+1)} \frac{1}{|B|} \He (V_j(B)\circ i_B) & \geq \epsilon_j \id_n \quad \text{for all $|\varphi| \geq \frac12 h_{j+1}$},
  \\
  \label{e:ind-ass2}
  L^{2j} \frac{1}{|B|} \He (V_j(B)\circ i_B) &\geq -O(g_j)\id_n \quad \text{for all $\varphi$}.
\end{align}
\end{lemma}

\begin{proof} 
For $|\varphi| \geq h_j$,
the estimate \eqref{e:ind-ass1} follows directly from the assumption \eqref{e:phi4-convex}
and the trivial bound $L^2\epsilon_j \geq \epsilon_j$.
Next we consider the case $\frac12 h_{j+1} \leq |\varphi| \leq h_j$. By definition,
\begin{equation}
    h_{j+1} = L^{-(j+1)} g_{j+1}^{-1/4} = L^{-(j+1)} g_j^{-1/4}(1+O(g_j)) = L^{-1}h_j (1+O(g_j)).
\end{equation}
Therefore \eqref{e:phi4-small-HeV} implies
\begin{equation}
  L^{2(j+1)} \frac{1}{|B|} \He (V(B) \circ i_B)
  \geq 
 (g_j(\frac12 L^{j+1} h_{j+1})^2  + \nu_j L^{2(j+1)} - O(g_j))
 \geq (\frac14 g_j^{1/2} - O(L^ 2 g_j))
  \geq \epsilon_j
  .
\end{equation}
Similarly, using Corollary~\ref{cor:phi4-small} for the small fields and the inductive
assumption for the large fields, we  have for all $\varphi$ that
\begin{equation}
  L^{2j} \frac{1}{|B|} \He (V_j(B)\circ i_B)
  \geq - O(g_j) \id_n,
\end{equation}
which implies \eqref{e:ind-ass2}.
This completes the proof of Lemma~\ref{lem:ind-ass}.
\end{proof}

The following proposition now advances the induction and thus proves Theorem~\ref{thm:phi4-convex}.

\begin{proposition} \label{prop:advance}
  Assume \eqref{e:ind-ass1}--\eqref{e:ind-ass2} with $j<N$.
  For $\varphi \in \R^n$ with $|\varphi| \geq h_{j+1}$ and $B_+ \in \mathcal{B}_{j+1}$,
  \begin{align}
    L^{2(j+1)} \frac{1}{|B_+|} \He (V_{j+1}(B_+) \circ i_{B_+})(\varphi) \geq (\epsilon_j-O(\vartheta_j^2\epsilon_j^2))\id_{n}
    .
  \end{align}
\end{proposition}

The proposition will be proved in the remainder of this section.
Since the scale $j$ will be fixed we usually drop the $j$ and write $+$ instead of $j+1$.
To set-up notation,
we fix a block $B_+ \in \mathcal{B}_{+}$ and write $V(B_+) = \sum_{B \in \cB_j(B_+)} V(B)$.
By the hierarchical structure, $\He V(B_+)$ is a block diagonal 
matrix indexed by the blocks $B \in \mathcal{B}(B_+)$,
and we will always restrict the domain to $X_j(B_+)$, the space of fields constant inside the small blocks $B$.
On this domain, $V(B_+)$ can be identified with a function of $L^d$ vector-valued variables while $V_+(B_+)$
has domain $X_{+}(B_+)$ and can be identified with a function of a single vector-valued variable.
The covariance operator $C$ and the projection $Q$ operate naturally on $X(B_+)=X_j(B_+)$
and can be identified with diagonal matrices indexed by blocks $B \in \cB(B_+)$;
in particular, they are invertible on $X(B_+)$.
By the definition of $V_+$ in \eqref{e:phi4-Vdef}, together with the hierarchical structure of $C$,
it follows that
\begin{align} \label{e:Wplus}
  V_+(B_+, \varphi) = -\log \E_{C}(e^{-V(B_+,\varphi+\zeta)})
  = -\log \int_{X(B_+)} e^{-H_\varphi(\zeta)} \, d\zeta + \text{constant},
\end{align}
where (recall that here $C$ denotes the restriction of $C$ to $X(B_+)$)
\begin{equation}
  H_\varphi(\zeta)
  = \frac12 (\zeta,C^{-1}\zeta) + V(B_+,\varphi+\zeta)
  .
\end{equation}
By differentiating \eqref{e:Wplus} we obtain, for $\dot\varphi \in X_+(B_+)$,
\begin{equation} \label{e:HeV+-def}
  \dot\varphi\He V_+(B_+,\varphi)\dot\varphi
  = \avg{\dot\varphi\He V(B_+,\varphi+\zeta)\dot\varphi}_{H_\varphi}
  - \var_{H_\varphi}(\nabla V(B_+,\varphi+\zeta) \cdot \dot\varphi)
\end{equation}
where $\avg{\cdot}_{H_\varphi}$ denotes the expectation of the probability measure with density $e^{-H_\varphi}$ on $X(B_+)$,
and $\nabla$ is the gradient in $X(B_+)$, i.e., with respect to fields that are constants on scale-$j$ blocks in $B_+$.

To estimate the right-hand side of the last equation,
we need some information on the typical value of the fluctuation field $\zeta$ under the expectation $\avg{\cdot}_{H_\varphi}$.
By assumption of the proposition, the bound \eqref{e:ind-ass2} holds, and together with the definition of $C = C_j$ in particular,
\begin{equation}
  \label{e:ind-ass2-bis}
  C^{1/2} \He V(B_+,\zeta)C^{1/2} \geq -\frac12 Q
  \quad \text{for all $\zeta \in X(B_+)$},
\end{equation}
as an operator on $X(B_+)$, i.e., $\zeta$ is a constant on every $B \in \mathcal{B}(B_+)$.
Therefore, uniformly in $\zeta$,
\begin{equation} \label{e:CHeHvarphi}
  C^{1/2}\He H_\varphi(B_+,\zeta)C^{1/2} = Q + C^{1/2} \He V(B_+,\varphi+\zeta) C^{1/2} \geq \frac12 Q.
\end{equation}
For any $\varphi$, the action $H_\varphi$ is therefore strictly convex  on $X(B_+)$ and, in particular, it has a unique minimiser in this space.
We denote this minimiser by $\zeta^0$. It satisfies the Euler--Lagrange equation
\begin{equation} \label{e:zeta-EL}
  \zeta^0 = - C \nabla V(B_+,\varphi+\zeta^0).
\end{equation}
Here recall the definition $V(B_+) = \sum_{B\in\mathcal{B}(B_+)} V(B)$,
and hence that $\nabla V(B_+)$ is a vector of blocks indexed by $B \in \mathcal{B}(B_+)$,
on which the covariance operator $C$ acts diagonally.

Further recall that $\varphi$ is constant on $B_+$.
By symmetry and uniqueness of the minimiser, we see that $\zeta^0$ has to be constant not only in each small block $B$,
but in each $B_+$, i.e., $\zeta^0 \in X_{+}(B_+)$.
In the following lemma, the block $B_+$ is fixed and $\varphi$ and $\zeta^0$ are both in $X_{+}(B_+)$
so that we may identify them with variables in $\R^n$.

\begin{lemma} \label{lem:zetamean}
  Let $|\varphi|\geq h_+$. Then $|\varphi+\zeta^0| \geq h_+(1 - O(g^{1/2}))$.
\end{lemma}
\begin{proof}
As discussed above, we regard $\nabla V$ and $C\nabla V$ both as block vectors indexed by $B \in\mathcal{B}(B_+)$.
For $\varphi'$ constant on $B_+$, the blocks of $\nabla V(B_+,\varphi')$ are equal
and $C$ acts by multiplying each of these blocks by the same constant $O(\vartheta^2 L^{2j})$.
Hence $C\nabla V(B_+,\varphi')$ is a block vector with all blocks equal to $O(\vartheta^2 L^{2j})\nabla V(B,\varphi')$
where $B$ is any of the block in $\mathcal{B}(B_+)$.
We denote by $|C\nabla V(B_+,\varphi')|_\infty$ the value in any of these blocks.
Now \eqref{e:phi4-small-dV} implies that, for $\varphi'$ constant on $B_+$ with $|\varphi'|\leq h_+$,
\begin{align} \label{e:Mbd}
  M &:= \sup_{|\varphi'|\leq h_+} |C \nabla V(B_+,\varphi')|_\infty
      \nnb
  &\leq
    \vartheta^2 L^{2j} \pB{ g h_+^3(1+O(g^{1/2})) + \nu h_+ + O(L^{-dj} h_+^{-1} g^{3/4})}
    \nnb
  &\leq
    \vartheta^2 h_+ \pa{g L^{2j} h_+^2(1+O(g^{1/2})) + L^{2j} \nu + O(L^{-2j} h_+^{-2} g^{3/4})}
   \leq
    O(\vartheta^2 g^{1/2}h_+)
    .
\end{align}
To prove the claim, we may
assume that $|\varphi+\zeta^0| \leq h_+$ since otherwise the claim holds trivially.
Then $|\zeta^0| \leq M = O(\vartheta^2 g^{1/2} h_+)$ by \eqref{e:zeta-EL} and \eqref{e:Mbd}.
We conclude from this that $|\varphi+\zeta^0| > h_+$ or $|\zeta^0| = O(\vartheta^2 g^{1/2}h_+)$.
Thus $|\varphi+\zeta^0| \geq h_+ \wedge (|\varphi|-O(\vartheta^2 g^{1/2}h_+)) \geq h_+(1- O(\vartheta^2 g^{1/2}))$.
\end{proof}

In the following lemma, $\zeta \in X(B_+)$ is the fluctuation field under the
measure with expectation $\avg{\cdot}_{H_\varphi}$. Thus $\zeta$ is constant in any
small block $B$, but unlike the minimiser $\zeta^0$ the field $\zeta$ is not constant in $B_+$.

\begin{lemma} \label{lem:zetabd}
  For any $t \geq 1$, with $\ell = L^{-j}$ as in \eqref{eq: scales},

  \begin{equation}
    \forall x \in B_+, \qquad      \P_{H_\varphi}(|\zeta_x-\zeta^0| \geq 3\vartheta\ell t) \leq 2e^{-t^2/4}.
  \end{equation}
\end{lemma}

\begin{proof}
  By changing variables, it suffices to study the measure with action
  $H(\zeta) = H_\varphi(\zeta+\zeta^0)$, whose unique minimiser is $\zeta=0$,
  and clearly $H$ has the same Hessian as $H_\varphi$.
  From the information that the minimiser of $H$ is $0$, we obtain a bound on the random variable $\zeta$ as follows.
  Using that $\He H \geq \frac12 C^{-1}$ as quadratic forms
  and that $C_{xx} \leq \vartheta^2\ell^2$ for all $x\in \Lambda$ by definition,
  the Brascamp--Lieb inequality \eqref{e:BL} for the measure $\avg{\cdot}_H$ with density proportional to $e^{-H}$ implies
  \begin{equation}
    \avg{e^{s(\zeta_x-\E_H(\zeta_x))}}_H
    \leq e^{s^2 C_{xx}}
    \leq e^{s^2 \vartheta^2 \ell^2}
    .
  \end{equation}
  By Markov's inequality therefore
  \begin{equation} \label{e:zetatail}
    \P_H(|\zeta_x-\avg{\zeta_x}_H| > \vartheta \ell t) \leq 2e^{-t^2/4}.
  \end{equation}
  To estimate the mean $\avg{\zeta}_H$, we integrate by parts to get
  \begin{equation}
      |B_+| \vartheta^2 \ell^2 \int e^{-H}
    \geq \sum_{x \in  B_+}C_{xx} \int e^{-H}
    = \int \pa{\nabla, C\zeta}  \, e^{-H}
    = \int (  \zeta,  C \nabla H(\zeta)) \, e^{-H} \geq \frac12 \int (\zeta,\zeta) \, e^{-H}
  \end{equation}
  where the integral is over $X(B_+)$ and $\nabla$ is the gradient on $X(B_+)$,
  and where we used that, by \eqref{e:CHeHvarphi},
  \begin{equation}
    (\zeta, C \nabla H(\zeta))
     = \int_0^1 (\zeta, C^{1/2} \He H(t\zeta) C^{1/2} \zeta) \, dt
    \geq \frac12(\zeta,\zeta).
  \end{equation}
  Since $\E(\zeta,\zeta) = | B_+| \avg{\zeta_x^2}_H$ by symmetry, therefore
  \begin{equation} \label{e:Ezeta}
    \avg{\zeta_x^2}_H \leq 2 \vartheta^2 \ell^2, \quad
    |\avg{\zeta_x}_H| \leq \sqrt{2} \vartheta \ell.
  \end{equation}
  Finally, combining \eqref{e:Ezeta} and \eqref{e:zetatail}
  \begin{equation}
    \P_H(|\zeta_x| > 3 \vartheta \ell t) \leq \P_H(|\zeta_x-\avg{\zeta_x}_H| \geq \vartheta \ell t) \leq 2e^{-t^2/4},
  \end{equation}
  which is the claim.
\end{proof}

Next we use the following estimate on $\He V_+(B_+)$.

\begin{lemma}
  Let $\varphi, \dot\varphi \in X_{+}(B_+)$. Then
  \begin{equation} \label{e:HeV+-bd}
    \dot\varphi \He V_+(B_+, \varphi) \dot\varphi
    \geq
    \avga{ \dot\varphi \frac{\He V(B_+, \varphi+\zeta)}{\id + C^{1/2} \He V(B_+,\varphi+\zeta) C^{1/2}} \dot\varphi}_{H_\varphi}
  \end{equation}
  where $\He V_+(B_+)$ is taken in $X_+(B_+)$ and $\He V(B_+)$ is taken in $X(B_+)$.
\end{lemma}

Note that $\He V(B_+,\varphi+\zeta)$ are both diagonal matrices
indexed by $B \in \mathcal{B}_+$, with constant entries on each block $B$.
In fact, $C$ is proportional to the identity matrix on $X(B_+)$.

\begin{proof}
We freeze the block spin field $\varphi \in X_{+}(B_+)$
and recall that the fluctuation field $\zeta \in X(B_+)$ is distributed with expectation $\avg{\cdot}_{H_\varphi}$.
We abbreviate $\He V = \He V(\varphi+\zeta) = \He V(B_+,\varphi+\zeta)$ throughout the proof.
Applying the Brascamp--Lieb inequality \eqref{e:BL0} to the measure $\avg{\cdot}_{H_\varphi}$ gives
\begin{equation}
  \var_{H_\varphi}(\nabla V(\varphi+\zeta)\cdot \dot\varphi)
  \leq \avg{\dot\varphi \He V(\varphi+\zeta) (C^{-1} + \He V(\varphi+\zeta) )^{-1} \He V(\varphi+\zeta) \dot\varphi}_{H_\varphi}
  .
\end{equation}
Inserting this into \eqref{e:HeV+-def}, the above can be written as
\begin{equation}
  \dot\varphi \He V_+(\varphi) \dot\varphi 
  \geq \avgB{\dot\varphi \pB{\He V - \He V (C^{-1} +  \He V )^{-1} \He V} \dot\varphi}_{H_\varphi}
  .
\end{equation}
Since $\He V$ and $C$ are both (block) diagonal matrices, the term inside the expectation can be written as
\begin{equation}
  \He V (\id +  C^{1/2} \He V  C^{1/2} )^{-1}.
\end{equation}
This completes the proof.
\end{proof}

For $\varphi \in X(B_+)$, let $\Lambda(\varphi)$ be the largest constant such that
$L^{2(j+1)} \He V(B_+,\varphi) \geq \Lambda(\varphi)$ 
as quadratic forms on $X(B_+)$.
From \eqref{e:ind-ass2-bis} it follows that $\Lambda(\varphi) \geq -\frac12$ uniformly in $\varphi \in X(B_+)$.
Then \eqref{e:HeV+-bd} implies that for $\dot\varphi \in X_{+}(B_+)$,
\begin{align}
  \dot\varphi \He V_+(B_+,\varphi) \dot\varphi
  &\geq
  L^{-2(j+1)}
  \avga{ \dot\varphi  \frac{ L^{2(j+1)} \He V(B_+, \varphi+\zeta)}{\id + C^{1/2} \He V(B_+,\varphi+\zeta) C^{1/2}} \dot\varphi}_{H_\varphi}
    \nnb
  &\geq
  L^{-2(j+1)} \avga{\frac{\Lambda(\varphi+\zeta)}{1+ L^{-2}\vartheta^2\Lambda(\varphi+\zeta)} }_{H_\varphi}
  (\dot\varphi,\dot\varphi),
\end{align}
where the second inequality uses that $t/(1+a t)$ is increasing in $t>-1/a$ and that $C \leq \vartheta^2 L^{2j}Q$.
\medskip

The next lemma completes the proof of Proposition~\ref{prop:advance}.

\begin{lemma}
For $\varphi \in X_{+}(B_+)$ with $|\varphi| \geq h_+$, we have
\begin{equation} \label{e:Lambdabd}
  \avga{\frac{\Lambda(\varphi+\zeta)}{1+L^{-2}\vartheta^2\Lambda(\varphi+\zeta)}}_{H_\varphi}
  \geq \epsilon-O(\vartheta^2 \epsilon^2).
\end{equation}
\end{lemma}

\begin{proof}
  On the event $\min_x |\varphi+\zeta_x| \geq \frac12 h_+$ we have $\Lambda(\varphi+\zeta) \geq \epsilon >0$ by \eqref{e:ind-ass1},
  and since $t/(1+at)$ is increasing for $t > 0$ therefore 
  \begin{equation} \label{e:Lambdabd-pf}
    \frac{\Lambda(\varphi+\zeta)}{1+L^{-2}\vartheta^2 \Lambda(\varphi+\zeta)} \geq \frac{\epsilon}{1+L^{-2}\vartheta^2 \epsilon}
    \geq \epsilon - O(\vartheta^2 \epsilon^2).
  \end{equation}
  By Lemma~\ref{lem:zetabd}, the probability that $|\zeta_x-\zeta^0| \geq \frac14 h_+$
  is bounded by $2e^{-(h_+/(12\vartheta\ell))^2/4} \leq 2e^{-c \, (\vartheta g)^{-1/2}}$ for any point $x\in B_+$
  (since $\vartheta\leq 1$).
  Using that $\zeta$ is constant on the small blocks $B$ and taking a union bound over the $L^d$ blocks $B \in \mathcal{B}(B_+)$
  we get that $\max_x |\zeta_x-\zeta^0| \geq \frac14 h_+$ with probability at most $2L^d e^{-c (\vartheta g)^{-1/2}}$.
  Since $|\varphi+\zeta^0| \geq h_+(1-O(g^{1/2})) \geq \frac34 h_+$ by Lemma~\ref{lem:zetamean},
  together with the assumption $|\varphi| \geq h_+$,
  we conclude that $\min_x|\varphi+\zeta_x| \geq \frac12 h_+$ with probability at least $1-2L^d e^{-c(\vartheta g)^{-1/2}}$.
  Thus \eqref{e:Lambdabd-pf} holds with at least this probability.
 
  On the event that \eqref{e:Lambdabd-pf} does not hold, we still have the bound $\Lambda(\varphi+\zeta) \geq -\frac12$ by \eqref{e:ind-ass2-bis}.
  Thus the contribution of this event to the expectation \eqref{e:Lambdabd} is bounded by $-O(L^d e^{-c \, (\vartheta g)^{-1/2}}) = -O(\vartheta^2\epsilon^4)$,
  where we used that $\epsilon_j \geq c \vartheta_j g_j$ by Lemma~\ref{lem:epsilon}.
  In summary,
  \begin{equation}
    \avga{\frac{\Lambda(\varphi+\zeta)}{1+\Lambda(\varphi+\zeta)}}_{H_\varphi}
    \geq \p{ \epsilon - O(\vartheta^2\epsilon^2) } (1-O(\vartheta^2\epsilon^4)) - O(\vartheta^2\epsilon^4)
    \geq \epsilon-O(\vartheta^2\epsilon^2).
  \end{equation}
  This implies the claim.
\end{proof}

\subsection{Proof of Theorem~\ref{thm:gap-phi4}}

We now use Corollary~\ref{cor:phi4-small} and Theorem~\ref{thm:phi4-convex}
to verify the assumptions of Corollaries~\ref{cor: sum D0}--\ref{cor: SG} and in doing so deduce
Theorem~\ref{thm:gap-phi4}.
By \eqref{e:DeltaHdecompQ}, the covariances in the decomposition of $(-\Delta_H+m^2)^{-1}$ are given by
\begin{equation}
  C_j = \lambda_j Q_j, \quad \text{with } \lambda_j =   L^{2j} 
  \begin{cases}
    O(1+m^2L^{2(j-1)})^{-2} &(j<N)\\
    O(1+m^2L^{2(N-1)})^{-1} &(j=N).
  \end{cases}
\end{equation}
We recall that $\vartheta_j = 2^{(j-j_m)_+}$.
Corollary~\ref{cor:phi4-small} implies
\begin{equation}
  \frac{1}{|B|}\He (V_j(B) \circ i_B) \geq (\nu_j + O(\vartheta_jL^{-2j}g_j^{5/4}))\id_n \quad \text{uniformly in $|\varphi| \leq h_j$.}
\end{equation}
The right-hand side is less than $0$ by Theorem~\ref{thm:phi4-rg}.
Thus, by Theorem~\ref{thm:phi4-convex},
the same estimate holds for $|\varphi|\geq h_j$ and therefore for all $\varphi$.
In summary, and since the above estimates hold for all blocks,
and using \eqref{e:HessiB},
\begin{equation} \label{e:phi4-assA1}
  C_j^{1/2} \He V_j(\varphi) C_j^{1/2} \geq L^{2j} \; (\nu_j + O(\vartheta_jL^{-2j}g_j^{5/4})) Q_j
  \quad \text{uniformly in $\varphi \in X_j$.}
\end{equation}
Thus Assumption~(A1) holds with
\begin{equation}
  \gep_j = (-L^{2j}\nu_j +O(\vartheta_jg_j^{5/4})).
\end{equation}

\begin{lemma} \label{lem:summu}
There exists a constant $\delta>0$ such that for all $j \in \N$,
\begin{equation} \label{e:summu}
  -2\sum_{k=0}^{j} L^{2k}\nu_k
  \leq \delta \frac{n+2}{n+8} \log g_{j} + O(1),
  \qquad
  \sum_{k=0}^\infty ((L^{2k}\nu_k)^2+\vartheta_k g_k^{5/4})
  = O(g_0^{1/4}).
\end{equation}
\end{lemma}

The elementary proof requires some notation from \cite{rg-brief};
we therefore postpone it to Appendix~\ref{app:phi4-rg-pf}.

\begin{proof}[Proof of Theorem~\ref{thm:gap-phi4}]
  We apply Corollary~\ref{cor: sum D0}.
  By \eqref{e:phi4-assA1}, Assumption~(A1) holds for all $j\leq N$, 
  and Assumptions~(A2) and~(A3) follow automatically from the hierarchical structure.
  Therefore,
  by \eqref{e:BL-rec-iterated}, the $|\varphi|^4$ measure satisfies a Brascamp--Lieb inequality with quadratic form
  \begin{equation} \label{e:phi4-G1}
    D_0
    \leq \sum_{j=0}^{N} \delta_j C_j , \quad
    \text{where } \delta_j = \exp\pa{2\sum_{k=1}^{j} \epsilon_k + O(\epsilon_k^2)}
    .
  \end{equation}
  We abbreviate $\gamma=(n+2)/(n+8)$.
  Using $g_j^{-1} = O(g_{j_m}^{-1})$ which holds by \eqref{e:gjm}, and using \eqref{e:summu},
  \begin{equation}
  \label{eq: 3.59}
    \exp\pa{2\sum_{k=1}^{j} \epsilon_k + O(\epsilon_k^2)}
    = O(g_{j_m}^{-\delta\gamma}).
  \end{equation}
  We then use that $g_{j_m}^{-1}= O(\log m^{-1})$ by \eqref{e:gjm}, 
  to show that \eqref{eq: 3.59} is a logarithmic correction of order 
  $(- \log m)^{\delta \gamma}$. Thus the dominant contribution in \eqref{e:phi4-G1}
  is given by 
  \begin{equation}
    \sum_{j=1}^{N-1} (1+m^2L^{2(j-1)})^{-2}L^{2j}
    + (1+m^2L^{2(N-1)})^{-1}L^{2N}
    = O(m^{-2}),
  \end{equation}
  where we recall that $m^2 \sim Ct (-\log t)^{-\gamma}$  as $t \downarrow 0$ by \eqref{e:phi4-rg-mepsilon}.
  In summary, we conclude that $D_0$ is bounded as a quadratic form from above by
  \begin{equation}
    O(m^{-2}) (\log m^{-1})^{\delta\gamma}
    =  O(t^{-1})(-\log t)^{(1+\delta)\gamma}.
  \end{equation}
  Replacing by $1+\delta$ by $\delta$,
  this implies the lower bound for the spectral gap claimed in \eqref{e:gap-phi4}.
  The upper bound for the spectral gap follows immediately from \eqref{e:phi4-varF}.
\end{proof}

\section{Hierarchical Sine-Gordon and Discrete Gaussian models}
\label{sec:sg}

In this section, we apply Corollaries~\ref{cor: sum D0}--\ref{cor: SG} to the hierarchical versions of the
Sine-Gordon and the Discrete Gaussian models.
This boils down to checking that Assumption (A1) is satisfied along the renormalisation group flow of 
both models.
Throughout this section $d=2$.

\subsection{Proof of Theorem~\ref{thm:gap-sg}} 

We start by defining the renormalisation group for the hierarchical Sine-Gordon model,
essentially in the set-up of \cite[Chapter 3]{MR2523458}.
By definition, with $\epsilon=\beta L^{-2N}$, the Sine-Gordon model has energy
\begin{equation} \label{e:sg-bis}
  H(\varphi) =
  \frac{\beta}{2}
  (\varphi,(-\Delta_H+L^{-2N} Q_N) \varphi) + \sum_{x\in\Lambda} V_0(\varphi_x),
\end{equation}
where the potential $V_0$ is even and $2\pi$-periodic.
We decompose the covariance of the Gaussian part as
\begin{equation} \label{e:DeltaHdecompP-eps}
  (-\beta\Delta_H + \beta L^{-2N} Q_N )^{-1}
  = \sum_{j=1}^N
  \beta^{-1}
  L^{2(j-1)} P_j + \beta^{-1}L^{2N}Q_N
  = \sum_{j=0}^N C_j
\end{equation}
with
\begin{gather}
  C_j = \lambda_j(\beta) Q_j,
  \qquad 
  \lambda_0(\beta) = \frac{1}{\beta},
  \quad
  \lambda_j(\beta)
  = \frac{\sigma}{\beta} L^{2j}
  \quad (0<j \leq N),
  \quad
  \sigma = 1-L^{-2}.
\end{gather}
Relative to this decomposition, the renormalised potential is defined as in
Section~\ref{sec:renorm-measure}.
Due to the hierarchical structure of this decomposition, the renormalised potential takes the form
\begin{equation}
  V_j(\varphi) = \sum_{B\in\mathcal{B}_j} V_j(B,\varphi),
\end{equation}
where $V_j(B,\varphi)$ only depends on $\varphi|_B$.
As in Section~\ref{sec:phi4-rg}, we restrict the domain of $V_j(B)$ to $X_j(B)$, i.e., the constant fields on $B$.
The final potential obtained as $V_{N+1}$ in \eqref{e:sg-V+} will instead be denoted by $V_{N,N}$
since it is indexed by the final block $\Lambda \in \cB_N$, i.e.,
$V_{N,N}(\varphi) = V_{N,N}(\Lambda_N,\varphi)$, and $\varphi$ can be seen as an external field.
Then each $V_j(B)$ can be identified with a $2\pi$-periodic function on $\R$
(and analogously for $V_{N,N}$).
For any such function $F: S^1 \to \R$, we use the norm
\begin{equation} \label{e:sg-norm}
  \|F\|= \sum_{q \in \Z} w(q) |\hat F(q)|, \quad w(q) = (1+|q|)^2,
\end{equation}
where our convention for the Fourier coefficients of $F$ is $\hat F(q) = (2\pi)^{-1} \int_0^{2\pi} F(\varphi) e^{iq\varphi} \, d\varphi$.
We write
\begin{equation}
  \|V_j\| = \|V_j(B)\| = \|V_j(B) \circ i_B\|,
  \qquad
  \hat V_j(0) = \hat V_j(B,0)
\end{equation}
for an arbitrary $B \in \mathcal{B}_j$ (the definition is independent of $B$).
Except for the weight $w(q)$, the norm \eqref{e:sg-norm} is the one used in \cite{MR2523458,MR1003504}.

\begin{proposition} 
  \label{prop:sg-rg}
  Let $j<N$. 
  Assume that $\|V_j- \hat V_j(0)\|$ is sufficiently small. Then the renormalised potential satisfies
  \begin{equation} \label{e:sg-rg}
    \|V_{j+1} - \hat V_{j+1}(0)\| \leq
    L^2 e^{-\sigma/2\beta}
    (\|V_j-\hat V_j(0)\| + O(\|V_j- \hat V_j(0)\|)^2).
  \end{equation}
  Moreover, for the last step $j=N$,
  \begin{equation} \label{e:sg-rg-last}
    \|V_{N,N} - \hat V_{N,N}(0)\| \leq
    \|V_N-\hat V_N(0)\| + O(\|V_N- \hat V_N(0)\|)^2.
  \end{equation}
\end{proposition}

The derivation of this proposition is postponed to Section \ref{sec: SG flow}.
We now state consequences of this proposition and prove Theorem~\ref{thm:gap-sg} using these.

\begin{corollary} \label{cor:sg-VNbd}
For every $\beta < \sigma/(4\log L)$ and $\kappa < L^{2} e^{-\sigma/2\beta} < 1$,
for all $V_0-\hat V_0$ sufficiently small,  
\begin{equation} \label{e:sg-Vjqbd}
  \|V_j-\hat V_j(0)\| \leq \kappa^j \|V_0-\hat V_0(0)\| \quad \text{for $j \leq N$,}
\end{equation}
and
\begin{equation} \label{e:sg-VNbd}
  \|V_{N,N}-\hat V_{N,N}(0)\|
  \leq
  2\kappa^N \|V_{0}-\hat V_{0}(0)\|.
\end{equation}
\end{corollary}

\begin{proof}
  Fix $\eta>0$ small and set $\kappa = L^2 e^{-(1-\eta)\sigma/2\beta} < 1$.
  The bound \eqref{e:sg-rg} implies that
  for $\|V_0-\hat V_0(0)\|$ sufficiently small depending on $\eta,\beta,\eta$,
  \begin{equation}
    \|V_{j+1} - \hat V_{j+1}(0)\|
    \leq L^2 e^{-(1-\eta)\sigma/2}   \|V_j-\hat V_j(0)\|
    = \kappa \|V_j-\hat V_j(0)\|.
  \end{equation}
  Then \eqref{e:sg-Vjqbd} follows by iterating this bound,
  and \eqref{e:sg-VNbd} follows from this and \eqref{e:sg-rg-last}.
\end{proof}

\begin{corollary}
\label{cor:sg-var}
Let $\beta < \sigma/(4 \log L)$ and let $\epsilon = \beta L^{-2N}$.
Then the variance of $F = \sum_{x \in \gL_N} \varphi_x$
under the Gibbs measure $\mu$ defined in \eqref{eq: mu mesure} is given by
\begin{equation} 
  \label{e:sg-var}
  \var_\mu (F)
  = \frac{|\Lambda_N|}{\epsilon} \p{1-O(\kappa^N)}
  .
\end{equation}
\end{corollary}

\begin{proof}
Throughout the following proof,
we denote by $C = (-\beta\Delta_H+\epsilon Q_N)^{-1}$ 
the full covariance of the hierarchical Gaussian free field.
By completion of the square,
and using that $(-\beta\Delta_H+\epsilon Q_N)\b 1 \epsilon^{-1} = \b 1$,
\begin{equation}
  -\frac12 (\varphi,(-\beta\Delta_H+\epsilon Q_N)\varphi) + t (\varphi, \b 1)
  = -\frac12 (\varphi-t \b 1 \epsilon^{-1}, (-\beta\Delta_H+\epsilon Q_N)(\varphi-t \b 1 \epsilon^{-1})) + \frac12 t^2 \epsilon^{-1}(\b 1,\b 1).
\end{equation}
With $F(\varphi) = \sum_x \varphi_x$, we get by translating the measure by $t\epsilon^{-1} \b 1$ that
\begin{align}
  \gG (t) = \log \E_C(e^{tF(\varphi)} e^{-V(\varphi)})
  = \frac12 t^2\epsilon^{-1} (\b 1,\b 1)
  + \log \E_C(e^{-V(\varphi+t \epsilon^{-1} \b 1)})
  = \frac{|\Lambda_N| t^2}{2\epsilon} -V_{N,N}(t \epsilon^{-1} \b 1).
\end{align}
By Corollary~\ref{cor:sg-VNbd} and the fact that the norm controls the second derivatives,
\begin{equation}
  | V_{N,N}''(0)| = |(V_{N,N}-\hat V_{N,N}(0))''| \leq \|V_{N,N}-\hat V_{N,N}(0)\|
  \leq 2\kappa^N \|V_0-\hat V_0(0)\|,
\end{equation}
where $V_{N,N}''$ is the second derivative of the function $V_{N,N}(\Lambda_N) \circ i_{\Lambda_N}: \R \to \R$.
Finally, and using that $\ddp{^2}{t^2} V_{N,N}(t \epsilon^{-1} \b 1) = V_{N,N}''(t \epsilon^{-1} \b 1) \epsilon^{-2}$
as well as that $\epsilon = \beta L^{-2N}$,
\begin{align}
  \var_\mu (F)
  = \ddp{^2 \gG (0)}{t^2}
  = \frac{|\Lambda_N|}{\epsilon} - \frac{V_{N,N}''(0)}{\epsilon^{2}}
  = \frac{|\Lambda_N|}{\epsilon} \pa{1- O\pa{\frac{\kappa^N}{\epsilon|\Lambda_N|}}}
  = \frac{|\Lambda_N|}{\epsilon} \p{1- O(\kappa^N)}
  .
\end{align}
This completes the proof.
\end{proof}

\begin{proof}[Proof of Theorem~\ref{thm:gap-sg}]
We start by proving the lower bound on the spectral gap by applying Corollary~\ref{cor: sum D0}.
Thanks to the hierarchical structure, the spins are constant in the blocks at any given scale $j$,
and Assumptions (A2) and (A3) always hold.
Assumption~(A1) follows from Corollary~\ref{cor:sg-VNbd} which implies that for $j \leq N$
\begin{equation}
 (V_j(B)\circ i_B)''(\varphi)
  \geq - \sum_q q^2 |\hat V_j(q)| = - \|V_j-\hat V_j(0)\| \geq -\kappa^j \|V_0-\hat V_0(0)\|.
\end{equation}
This implies the bound \eqref{e:HeVB borne inf} with 
\begin{equation}
s = \frac{1}{ |B_j|} \; \kappa^j \|V_0-\hat V_0(0)\| = \kappa^j  \|V_0-\hat V_0(0)\| \; L^{-2 j}.
\end{equation}
The equivalent of \eqref{e:HessiB} is
\begin{equation} 
\label{e:HessiB-sg}
C_j^{1/2}(\He_{X_j} V_j)C_j^{1/2}
\geq - s L^{2j}  Q_j .
\end{equation}
Therefore Assumption~(A1) in \eqref{e:ass-V} holds
with $\epsilon_j = s L^{2j} =  \kappa^j \|V_0-\hat V_0(0)\|$.
With $\delta_j$ defined as in \eqref{e:BL-rec-iterated}, it follows that
\begin{align} 
  \sum_{j=0}^{N} \delta_j C_j
  &\leq \exp\pa{\sum_{j=0}^{N} O(\kappa^j)\|V_0-\hat V_0(0)\|} \sum_{j=0}^{N} C_j 
    \nnb
  &\leq (1+O(\|V_0-\hat V_0(0)\|)) (-\beta\Delta_H+\epsilon Q_N)^{-1}
    \leq \frac{O(1)}{\epsilon} \id_{\gL_N}.
\end{align}
Applying Corollary~\ref{cor: sum D0}, we get that
the measure $\mu$ satisfies a Brascamp-Lieb inequality with matrix
\begin{equation}
  D_0 \leq \frac{O(1)}{\epsilon} \id_{\Lambda_N}.
\end{equation}
This implies immediately the asserted lower bound on the spectral gap,
i.e., $\gamma_N \geq c\epsilon$.

Finally, the upper bound on the spectral gap follows readily from  Corollary \ref{cor:sg-var}.
Choosing as test function  $F = \sum_{x \in \gL_N} \varphi_x$,
we have $\bbE_\mu(\nabla F, \nabla F) = |\Lambda_N|$ and \eqref{e:sg-var} implies
\begin{equation}
  \frac{\bbE_\mu(\nabla F, \nabla F) }{\var_\mu (F)}
  = \epsilon(1+O(\frac{1}{\epsilon L^{2N}})) = O(\epsilon).
\end{equation}
This completes the proof.
\end{proof}

\subsection{Proof of Proposition~\ref{prop:sg-rg}}
\label{sec: SG flow}

The proof of Proposition~\ref{prop:sg-rg} follows as in \cite[Chapter~3]{MR2523458},
with small modifications.
Throughout Section~\ref{sec: SG flow}, the full covariance matrix $(-\beta\Delta_H+\epsilon Q_N)^{-1}$
does not play a role and we write $C = C_j$ for a fixed scale $j$.
More generally, we drop the scale index $j$ and write $+$ in place of $j+1$.
We write $B_+$ for a fixed block in $\cB_+$ and $B$ for the blocks in $\cB(B_+)$.

We need the following properties of the norm \eqref{e:sg-norm}.
Since $w(p+q) \leq w(p)w(q)$, i.e.,
\begin{align}
  (1+|p+q|)^2
  &= 1+p^2+q^2 +2|p+q|+2pq
  \nonumber\\
  &\leq
  1+p^2+q^2+2|p+q|+4|pq|+2|pq|(|p|+|q|)
  = (1+|p|)^2(1+|q|)^2,
\end{align}
the norm \eqref{e:sg-norm} satisfies the product property
\begin{equation}
  \|FG\| = \sum_{q,p} w(q) |\hat F(q-p)||\hat G(p)| \leq \sum_{q,p} w(q-p)w(p) |\hat F(q-p)||\hat G(p)|
  = \|F\|\|G\|.
\end{equation}
As a consequence, for any $F: S^1 \to \R$ with $\|F\|$ small enough,
\begin{align}
  \|e^{-F}-1\| &\leq \|F\| + O(\|F\|^2), \label{e:normbd1}\\
  \|\log(1+F)\| &\leq \|F\|+ O(\|F\|^2). \label{e:normbd3}
\end{align}

\begin{lemma}
For $F:S^1 \to \R$ with $\hat F(0) = 0$ and $\|F\|<\infty$, and for $x\in\Lambda$,
\begin{equation} \label{e:Exbd}
  \|\E_C \left( F(\cdot+\zeta_x) \right)\|
  \leq e^{-\sigma/(2\beta)} \|F\|.
\end{equation}
\end{lemma}

\begin{proof}
By \eqref{e:DeltaHdecompQ}, under the expectation $\E_C$,
each $\zeta_x$ is a Gaussian random variable with variance $\sigma/\beta$. Therefore 
\begin{equation}
  \E_C(e^{iq\zeta_x})
  = e^{-\sigma q^2/(2\beta)}.
\end{equation}
This gives
\begin{equation}
  \E_C (F(\varphi+\zeta_x)) =
  \E_C \qbb{ \sum_q \hat F(q) e^{iq(\varphi+\zeta_x)} }
  = \sum_q e^{-\sigma q^2/(2\beta)} \hat F(q) e^{iq\varphi}.
\end{equation}
Since by assumption $\hat F(0) = 0$, we obtain
\begin{equation}
  \|  \E_C (F(\cdot+\zeta_x))\|
  \leq \sum_q e^{-\sigma q^2/(2\beta)} w(q) |\hat F(q)|
  \leq e^{-\sigma/(2\beta)} \sum_q w(q) |\hat F(q)|
  = e^{-\sigma/(2\beta)} \|F\|
\end{equation}
as claimed.
\end{proof}

\begin{proof}[Proof of Proposition~\ref{prop:sg-rg}]
We may assume that $\hat V(0)=0$.
We fix $B_+ \in \cB_+$ and use $B$ for the blocks in $\cB(B_+)$.
By definition of the hierarchical model, the Gaussian field $\zeta$ with covariance $C=C_j$
is constant in any block $B \in \cB_j$ and we thus write $\zeta_B$ for $\zeta_x$ with $x\in B$.
We then start from
\begin{align}
  e^{-V_+(B_+,\varphi)} = \E_{C}  \left( \prod_{B\in \mathcal{B}(B_+)} e^{-V(\varphi+\zeta_B)} \right)
  &= \E_{C}  \left( \prod_{B \in \mathcal{B}(B_+)} (1+e^{-V(\varphi+\zeta_B)}-1) \right)
    \nnb
  &= \sum_{X \subset B_+} \E_{C}  \left( \prod_{B \in \mathcal{B}(X)}(e^{-V(\varphi+\zeta_B)}-1) \right),
\end{align}
where $X\subset B_+$ denotes that $X$ is a union of blocks $B\in \mathcal{B}(B_+)$.
The term with $|X|=0$ is simply $1$.
By \eqref{e:Exbd} and \eqref{e:normbd1}, the terms with $|X|=1$ are bounded by
\begin{equation}
 \left \|\sum_{B\in \mathcal{B}(B_+)} \E_{C} \Big( e^{-V(\varphi+\zeta_B)}-1 \Big)  \right \|
 \leq |\mathcal{B}(B_+)| e^{-\sigma/(2\beta)}
 (\|V\|+O(\|V\|^2))
  .
\end{equation}
By \eqref{e:Exbd}, using that the $\zeta_B$ are independent for different blocks $B$ and the product property of the norm,
the terms with $|X|>1$ give
\begin{align}
 \left \|   \sum_{|X| > 1} \E_{C}  \left( \prod_{B\in \mathcal{B}(X)}(e^{-V(\varphi+\zeta_B)}-1) \right)  \right\|
  &\leq \sum_{|X|>1} \prod_{B \in \mathcal{B}(X)}e^{-\sigma/(2\beta)} \|(e^{-V(\varphi+\zeta_B)}-1) \|
    \nonumber\\
  &\leq \sum_{|X|>1} (e^{-\sigma/(2\beta)}(\|V\|+O(\|V\|^2)))^{|X|}
  = O(e^{-\sigma/(2\beta)}\|V\|^2).
\end{align}
In summary, for $\| V\|$ small enough, we get 
\begin{align}
 \left \| \E_{C} \left( \prod_{B\in \mathcal{B}(B_+)} e^{-V(\varphi+\zeta_B)} \right) -1   \right\|
  &\leq |\mathcal{B}(B_+)| e^{-\sigma/(2\beta)}
    (\|V\| + O(\|V\|^2))
    \nnb
  &= L^2 e^{-\sigma/(2\beta)}
  (\|V\|+O(\|V\|^2)).
\end{align}
Finally, by \eqref{e:normbd3},
\begin{equation}
  \|V_+\| = \left\| \log \left(1 + \E_C \left( \prod_{B\in \mathcal{B}(B_+)} e^{-V(\varphi+\zeta_B)}  \right) - 1\right) \right\|
  \leq L^2 e^{-\sigma/(2\beta)}
  (\|V\|+O(\|V\|^2)),
\end{equation}
as needed.
\end{proof}

\subsection{Proof of Theorem~\ref{thm:gap-dg}}

We will now reduce the result for the Discrete Gaussian model to that for the Sine-Gordon model.
For this, we carry out an initial renormalisation group step by hand, resulting in an effective
Sine-Gordon potential for the Discrete Gaussian model.
This strategy for the Discrete Gaussian model (and more general models) goes back to \cite{MR634447}.

First, recall that the covariance of the hierarchical GFF can be written as
\begin{equation}
  (-\beta\Delta+\epsilon Q_N)^{-1}
  = C_0 + \cdots + C_N = C_0 + C_{\geq 1},
\end{equation}
where $C_0 = \frac{1}{\beta} Q_0$ and
where $Q_0$ is simply the identity matrix on $\R^\Lambda$.
Therefore, by the convolution property of Gaussian measures, 
\begin{equation}
  e^{-\frac{1}{2} (\sigma, (-\beta\Delta_H+\epsilon Q_N) \sigma)}
  \propto \int_{\R^{\Lambda}} e^{-\frac{1}{2} (\varphi,C_{\geq 1}^{-1}\varphi)} 
  e^{- \frac{\beta}{2}(\varphi-\sigma,\varphi- \sigma)} \, d\varphi
  \propto \E_{C_{\geq 1}}(e^{- \frac{\beta}{2} (\varphi-\sigma,\varphi-\sigma)}),
\end{equation}
where $A \propto B$ denotes that $A/B$ is independent of $\sigma$,
and where the Gaussian expectation applies to the field $\varphi$.
We define the effective single-site potential $V(\psi)$ for $\psi\in \R$ by
\begin{equation} \label{e:V-DG}
  e^{-V(\psi)} = \sum_{n \in 2\pi\Z} e^{-\beta (n-\psi)^2/2}.
\end{equation}
The potential $V$ is $2\pi$-periodic as in the Sine-Gordon model.
This is where the $2\pi$-periodicity of the Discrete Gaussian Model is convenient.
For $\psi\in\R$, we also define a probability measure $\mu_\psi$ on $2\pi\Z$ by
\begin{equation} 
\label{e:mupsi}
  \mu_\psi(n) = e^{V(\psi)} e^{-\beta(n-\psi)^2/2} \quad \text{for $n \in 2\pi\Z$.}
\end{equation}
For $\varphi \in \R^\Lambda$, we further set $\mu_\varphi = \prod_{x\in\Lambda} \mu_{\varphi_x}$
with $\mu_{\varphi_x}$ as in \eqref{e:mupsi} with $\psi=\varphi_x$. 
With this notation, in summary, we have the representation
\begin{equation}
  \sum_{\sigma \in (2\pi\Z)^\Lambda} F(\sigma) \, e^{-\frac{1}{2} (\sigma,(-\beta\Delta_H+\epsilon Q_N) \sigma)}
  \propto
  \E_{C_{\geq 1}}(e^{- V(\varphi)} \E_{\mu_\varphi}(F(\sigma)))
  .
\end{equation}

Denote by $\mu_r(d\varphi)$ the probability measure on $\R^\Lambda$ of the Sine-Gordon model
with potential $V(\varphi)$ defined by \eqref{e:V-DG} with $C_{\geq 0}$ replaced by $C_{\geq 1}$.
\begin{equation} \label{e:DG-E-decomp}
  \E_\mu(F) = \E_{\mu_r}(\E_{\mu_\varphi}(F)).
\end{equation}

In the next two lemmas, we verify
that $V$ satisfies the conditions of  Theorem~\ref{thm:gap-sg} 
provided $\beta$ is sufficiently small,
and
that the probability measure $\mu_\psi$ satisfies a spectral gap inequality on $2\pi\Z$, with constant uniform in $\psi$.
It is clear from the definition \eqref{e:V-DG} that $V$ is $2\pi$-periodic.

\begin{lemma} \label{lem:Vr}
  For $\beta>0$ small enough, $V$ is smooth with $\|V-\hat V(0)\|=O(e^{-1/(2\beta)})$.
\end{lemma}

\begin{proof}
  The function $F = e^{-V}$ is $2\pi$-periodic, and subtracting a constant from $V$, we can normalise
  $F$ such that $\hat F(0)=1$. Note that subtraction of a constant does not change $V-\hat V(0)$.
  The Fourier coefficients of $F$ are then given by
  \begin{equation}
    \hat F(q)
    = \frac{1}{2\pi} \int_0^{2\pi} F(\psi) e^{-iq\psi} \, d\psi
    = \frac{C}{2\pi} \int_\R e^{-\beta \psi^2/2} e^{-iq\psi} \, d\psi
    = e^{-q^2/(2\beta)}
    ,
  \end{equation}
  where the constant $C$ and the last equality are due to the normalisation $\hat F(0)=1$.
  It follows that
  \begin{equation}
    \|F-1\| = \sum_{q\neq 0} (1+q^2) e^{-q^2/(2\beta)}
    = O(e^{-1/(2\beta)})
    .
  \end{equation}
  By \eqref{e:normbd3}, it then also follows that
  \begin{equation}
    \|V\| = \|\log F\| = \|\log (1+(F-1))\| = \|F-1\| + O(\|F-1\|^2) = O(e^{-1/(2\beta)}).
  \end{equation}
  Since $\|V-\hat V(0)\| \leq \|V\|$, this clearly implies the claim.
\end{proof}

\begin{corollary} \label{cor:mur}
  For $\beta>0$ sufficiently small,
  the measure $\mu_r$ has inverse spectral gap $O(1/\epsilon)$.
\end{corollary}

\begin{proof}
  The proof is essentially the same as that of Theorem~\ref{thm:gap-sg}.
  The only difference compared to Theorem~\ref{thm:gap-sg} is that we replaced $C_{\geq 0}$ by $C_{\geq 1}$
  which does not change the conclusion.
  For small $\beta$,
  the assumption on $V$ is satisfied thanks to Lemma~\ref{lem:Vr}.
\end{proof}

The following lemma can be proved, e.g., using the \emph{path method} for spectral gap inequalities;
we postpone the elementary proof to Appendix~\ref{app:singlespin-dg}.

\begin{lemma} 
\label{lem:gap-dg-single}
  For any $\beta>0$, there exists a constant $C_\beta$ such that
  the measure $\mu_\psi$ on $2\pi\Z$ has a spectral gap uniformly in $\psi \in \R$,
  \begin{equation} 
  \label{e:gap-dg-single}
    \var_{\mu_\psi}(F(n)) \leq C_\beta 
 \E_{\mu_\psi} \Big( (F(n+2\pi)-F(n))^2 + (F(n-2\pi)-F(n))^2  \Big)   .
  \end{equation}
\end{lemma}

With the above ingredients, the proof can now be completed as follows.

\begin{proof}[Proof of Theorem~\ref{thm:gap-dg}]
We start with the proof of the lower bound on the spectral gap.
By \eqref{e:DG-E-decomp}, the variance of a function  $F: (2\pi\Z)^\Lambda \to \R$ under the Discrete Gaussian measure can be written as
\begin{equation} \label{e:DG-var-decomp}
  \var_\mu(F) = \E_{\mu_r}(\var_{\mu_\varphi}(F)) + \var_{\mu_r}(G), \quad \text{where } G(\varphi) = \E_{\mu_\varphi}(F).
\end{equation}
By Corollary~\ref{cor:mur}, the measure $\mu_r$ has an inverse spectral gap bounded by $O(1/\epsilon)$.
By Lemma~\ref{lem:gap-dg-single} and the tensorisation principle
for spectral gaps, the product measure $\mu_\varphi = \prod_{x\in\Lambda} \mu_{\varphi_x}$ has a spectral gap uniformly bounded by $C_\beta$.
It follows that
\begin{equation}
  \var_\mu(F) \leq
  C_\beta  \bbD (F)
  +
  O(\frac{1}{\epsilon})
  \sum_{x\in\Lambda} \E_{\mu_r}(|\nabla_{\varphi_x} G|^2),
\end{equation}
where the Dirichlet form introduced in \eqref{eq: dirichlet discrete}
has been denoted by 
\begin{equation}
  \bbD (F) =
  \frac{1}{2(2\pi)^2}
  \sum_{x\in\Lambda} \bbE_\mu  
  \Big( (F(\sigma^{x+})-F(\sigma))^2  + (F(\sigma^{x-})-F(\sigma))^2 \Big).
\end{equation}

We also set 
\begin{equation}
  \bbD_{x,\mu_{\varphi}} (F) = 
  \frac{1}{2(2\pi)^2}
  \bbE_{\mu_{\varphi}}
  \Big( (F(\sigma^{x+})-F(\sigma))^2  + (F(\sigma^{x-})-F(\sigma))^2 \Big).
\end{equation}
Then the second term on the right-hand side is bounded as follows.
Since with respect to the measure $\mu_\varphi$ for fixed $\varphi$,
the $\sigma_x$ are independent, we have
\begin{equation}
  |\nabla_{\varphi_x} G(\varphi)|^2
  = \beta^2  (\cov_{\mu_{\varphi}}(F(\sigma), \sigma_x))^2
  \leq \beta^2 \E_{\mu_{\varphi}}(\cov_{\mu_{\varphi_x}}(F(\sigma), \sigma_x))^2)
  \leq C_\beta^2 
  \bbD_{x,\mu_{\varphi}} (F)
\end{equation}
where we used the following inequality,
which follows from $\var_{\mu_{\varphi_x}}(\sigma_x) \leq C_\beta$ and \eqref{e:gap-dg-single}:
\begin{equation}
  (\cov_{\mu_{\varphi_x}}(F(\sigma),\sigma_x))^2
  \leq
  \p{\var_{\mu_{\varphi_x}}(F)} \p{\var_{\mu_{\varphi_x}}(\sigma_x)}
  \leq C_\beta^2  \bbD_{x, \mu_{\varphi}} (F)
  .
\end{equation}
Using that $\bbD(F) = \sum_{x\in\Lambda} \E_{\mu_r}(\bbD_{x,\mu_\varphi}(F))$,
in summary, we conclude that
\begin{equation}
  \var_\mu(F) \leq
  C_\beta \Big(1+C_\beta O(\frac{1}{\epsilon}) \Big) \bbD  (F)
\end{equation}
and therefore that the inverse spectral gap obeys $1/\gamma = O(1/\epsilon)$.

For the matching upper bound on the spectral gap,
we use the test function $F = \sum_{x\in \Lambda} \sigma_x$,
analogously to the Sine-Gordon case.
For any $\psi \in \R$ and $t \in \R$,
\begin{equation}
  \E_{\mu_\psi}(e^{t\sigma})
  = e^{V(\psi)} \sum_{n\in 2\pi\Z} e^{-\beta(n-\psi)^2/2 +  n t}
  = e^{V(\psi)-V(\psi+ t/\beta) +t^2/(2\beta) + t\psi}.
\end{equation}
Let $u = \sum_{y} [C_{\geq 1}]_{xy}$ (which is independent of $x$).
It follows that
\begin{align}
  e^{\Gamma(t)}
  = \E_{\mu}(e^{tF})
  = \E_{\mu_r} \E_{\mu_\varphi}(e^{tF})
  &=
  e^{t^2 |\Lambda_N|/(2\beta)} 
  \frac{\E_{C_{\geq 1}}(e^{-\sum_x V(\varphi_x+ t/\beta) + t \sum_x \varphi_x})}{\E_{C_{\geq 1}}(e^{-V(\varphi)})}
  \nnb
  &=
    e^{t^2 |\Lambda_N| (1/\beta+u)/2} \frac{\E_{C_{\geq 1}}
    (e^{-\sum_x V(\varphi_x+ t/\beta+ t  u)})}{\E_{C_{\geq 1}}(e^{-V(\varphi)})}
  .
\end{align}
Since $\sum_y [C_0]_{xy} = [C_0]_{xx}=1/\beta$,
note that $1 /\beta + u = \sum_y \sum_{j=0}^N [C_j]_{xy}= \sum_{y} (-\beta\Delta_H+\epsilon Q_N)^{-1}_{xy} =   \epsilon^{-1}$.
As in the proof of Corollary \ref{cor:sg-var}, it follows that
\begin{equation}
  \var_{\mu}(F) = \frac{|\Lambda_N|}{\epsilon} - \frac{V_N''(0)}{\epsilon^2}
  = \frac{|\Lambda_N|}{\epsilon}(1+O(\frac{\kappa^N}{\epsilon L^{2N}}))
  = \frac{|\Lambda_N|}{\epsilon}(1+O(\kappa^N)).
\end{equation}
Since $\bbD(F) = |\Lambda_N|$, this completes the proof of $\gamma \leq \epsilon(1+O(\kappa^N))$
and therefore the proof of the theorem.
\end{proof}

\appendix
\section{Estimates for log-concave measures}
\label{app:convex}

Let $X$ be a finite-dimensional vector space with inner product $(\cdot,\cdot)$
and Lebesgue measure $m$.
Choosing an orthonormal basis, we may identify $X$ with $\R^k$ for some $k$.
Using this identification or the inner product structure directly,
the gradient, Laplacian, and Hessian of a function $F: X \to \R$ are defined.
Assume that $H: X \to \R$
satisfies $\He H > c \id$ uniformly for a constant $c>0$.
Let $\mu$ be the probability measure on $X$ with density proportional to $e^{-H}$ with respect to $m$.
Let $L$ be the (positive) self-adjoint generator
of the Langevin dynamics leaving $\mu$ invariant, i.e., for smooth $F: X \to \R$,
\begin{equation}
L F (\gz) = -\Delta F (\gz) + (\nabla H(\zeta),  \nabla F (\gz)),
\end{equation}
where $\nabla$ and $\Delta$ are the gradient and Laplacian on $X$.

In Sections~\ref{sec:recursion} and~\ref{sec:phi4},
we make use of the Helffer--Sj\"ostrand representation and the Brascamp--Lieb inequality.
Define the operator $\cL$ (Witten Laplacian) on $D \otimes X \subset L^2(\mu) \otimes X$ by
\begin{equation}
\cL = L \otimes \id +   \He H \, , 
\end{equation}
where $D \subset L^2(\mu)$ is the domain on which the operator $L$ is self-adjoint.
Then one has the Helffer-Sj\"ostrand representation \cite{MR1257821} (see also \cite{MR1936110})
for the covariance of two random variables $F,G: X \to \R$.

\begin{theorem}[Helffer--Sj\"ostrand representation] \label{thm:HS}
For sufficiently smooth $F,G: X \to \R$,
\begin{equation} \label{e:HS}
  \cov_{\mu}(F,G)
  = \E_{\mu}( \nabla F, \cL^{-1}  \, \nabla  G).
\end{equation}
\end{theorem}

In particular, one can easily obtain the Brascamp--Lieb inequality \cite{MR0450480}
from this representation.

\begin{theorem}[Brascamp--Lieb inequality] \label{thm:BL}
  For sufficiently nice $F: X \to R$,
  \begin{equation} \label{e:BL0}
    \var_\mu(F) \leq \E_\mu(\nabla F, (\He H)^{-1} \nabla F).
  \end{equation}
  In particular, if $\He H(\varphi) \geq Q$ uniformly in $\varphi \in X$, then for any $f\in X$,
  \begin{equation} \label{e:BL}
    \log \E_\mu(e^{(f,\zeta) - \E_\mu(f,\zeta)}) \leq \frac12 (f,Q^{-1} f).
  \end{equation}
\end{theorem}

\section{Proof of Theorem~\ref{thm:phi4-rg} and of Lemmas~\ref{lem:epsilon} and \ref{lem:summu}}
\label{app:phi4-rg-pf}

In this appendix,
we translate the results from \cite{rg-brief} to assume the form stated in Theorem~\ref{thm:phi4-rg},
and we proof two elementary lemmas for the sequences $(\epsilon_j)$ and $(\mu_j)$.

\begin{proof}[Proof of Theorem~\ref{thm:phi4-rg}]
  First, since our constants are allowed to depend on $L$
  and since we are only considering derivatives of finite order, the constants $\ell_0$ and $k_0$ that
  appear in the definitions of the versions of $\ell_j$ and $h_j$ in \cite{rg-brief} are insignificant
  for our estimates here and we therefore drop them.

  The critical point $\nu_c =  \nu_c(g)$ is chosen as in \cite[Theorem~4.2.1]{rg-brief}.
  Moreover, given $t = \nu-\nu_c(g) > 0$ and $g=g_0>0$ small,
  the mass parameter $m^2>0$ and $\nu_0= \nu-m^2$ are determined as in the proof of \cite[Theorem~4.2.1]{rg-brief}.
  In \cite{rg-brief}, the renormalisation group flow is defined in terms of the decomposition
  of $(-\Delta_H+m^2)^{-1}$ in terms of the orthogonal projections $P_j$ as in \eqref{e:DeltaHdecompP}, namely as
  \begin{equation}
    (-\Delta_H+m^2)^{-1} = \sum_{j=1}^N \tilde\lambda_j P_j + \frac{1}{m^2} Q_N = \sum_{j=1}^N \tilde C_j + \hat C_N,
    \quad
    \tilde \lambda_j = \frac{L^{2(j-1)}}{1+L^{2(j-1)}m^2},
  \end{equation}
  where we here write $\tilde C_j = \tilde\lambda_j P_j$ for the covariances denoted by $C_j$ in \cite{rg-brief}
  to distinguish them from the covariances $C_j = \lambda_jQ_j$ that we primarily use in this paper.
  In terms of these, we also have
  \begin{equation}
    (-\Delta_H+m^2)^{-1} = \sum_{j=0}^N \lambda_j Q_j.
  \end{equation}
  To translate between the two decompositions, note that $\sum_{j=0}^k \lambda_j Q_j =\sum_{j=1}^k \tilde\lambda_j P_j + \tilde\lambda_{j+1} Q_j$,
  i.e.,
  \begin{equation}
    \sum_{j=0}^k C_j = \sum_{j=1}^k \tilde C_j + \tilde \lambda_{j+1} Q_j.
  \end{equation}
  
  In \cite[Theorem~6.2.1]{rg-brief}, 
  it is shown that there is a sequence $(\tilde g_j,\tilde\nu_j,\tilde u_j)$ and a sequence of functions $\tilde K_j$,
  with $\tilde g_0 = g_0 = g$, $\tilde\nu_0 = \nu_0$, $\tilde u_0=0$ and $\tilde K_0=0$,
  and with estimates as stated in that proposition, such that
  \begin{equation}
    \E_{\tilde C_1 + \cdots + \tilde C_j}(e^{-V_0(\Lambda, \varphi+\zeta)}) = e^{-\tilde u_j|\Lambda|} \prod_{B \in \cB_j} (e^{-\tilde V_j(B,\varphi)} + \tilde K_j(B,\varphi)),
  \end{equation}
  where again we use tildes to refer to the quantities as defined in \cite{rg-brief}.
  Therefore, using the relation between the two decompositions $(C_j)$ and $(\tilde C_j)$,
  our effective potential $V_j$ as defined in \eqref{e:phi4-Vdef} in terms of the decomposition $(C_j)$ is given by
  \begin{equation}
    e^{-V_{j}(\varphi)}
    = \prod_{B\in\mathcal{B}_j} e^{-V_j(B, \varphi)}
    = e^{-\tilde u_j|\Lambda|} \E_{\tilde\lambda_{j+1}Q_j} \qa{ \prod_{B \in \mathcal{B}_j} (e^{-\tilde V_j(B,\varphi)} + \tilde K_j(B,\varphi))}.
  \end{equation}
  Differently from the usual renormalisation group steps, the expectation on the right-hand side does not involve any reblocking,
  i.e., the size of the blocks is the same on both sides of the equality.
  This is the same situation as in the last renormalisation group step in
  \cite[Proposition~6.2.2]{rg-brief}. In \cite{rg-brief}, the last renormalisation group step is only applied at the last
  scale, but it we can here apply it at any scale.
  More precisely,
  by \cite[Proposition~6.2.2 and Remark~10.7.2]{rg-brief} with the covariance $\hat C$ replaced by $\hat C = \tilde\lambda_{j+1} Q_j$
  and the scale $N$ replaced by $j$,
  we obtain
  \begin{equation}
    e^{-V_{j}(\varphi)} = e^{-\hat u_j |\Lambda|} \prod_B \pB{ e^{-\hat V_j(B,\varphi)}(1+\hat W(B,\varphi)) + \hat K_j(B,\varphi)},
  \end{equation}
  as in \eqref{e:phi4-rg-repr}.
  Moreover, the bounds on the $T_\infty(h)$-norm and the $T_0(\ell)$-norm of $\hat K$
  stated in \cite[Proposition~6.2.2 and Remark~10.7.2]{rg-brief} directly translate
  directly to the estimates \eqref{e:phi4-rg-Tinfty} and \eqref{e:phi4-rg-T0}.

  Finally, \eqref{e:phi4-rg-mepsilon} is a consequence of \cite[Corollary~6.2.2]{rg-brief}, together
  with the definition of $\nu_c(g)$ below \cite[(6.2.24)]{rg-brief} and \cite[Theorem~4.2.1]{rg-brief}
  for the asymptotics of $m^2 = 1/\chi$.
\end{proof}

Finally, we prove the elementary estimates
for the sequences $\epsilon_j$ and $\mu_j$ stated in Lemmas~\ref{lem:epsilon} and \ref{lem:summu}.

\begin{proof}[Proof of Lemma~\ref{lem:epsilon}]
  By decreasing $\epsilon_j$ to $\frac15 g_j^{1/2}$ if necessary,
  we can assume that $\epsilon_{j+1} = \epsilon_j - \gamma_j\epsilon_j^2$ with $\gamma_j=O(\vartheta_j^2)$.
  To obtain an lower bound on the sequence $(\epsilon_j)$, we may also increase the $\gamma_j$
  and assume that $\gamma_j = \gamma \vartheta_j^2$ for some $\gamma=O(1)$.
  The solution to this recursion behaves as
  \begin{equation}
    \epsilon_j
    \asymp
    \frac{\epsilon_0}{1+\epsilon_0 \sum_{k\leq j} \gamma_k}
    \asymp
    \frac{\epsilon_0}{1+\epsilon_0 \gamma (j \wedge j_m)}.
  \end{equation}
  This follows, e.g., from \cite[Proposition~6.1.3 and (6.1.9)]{rg-brief}. Likewise, the sequence $g_j$ obeys
  \begin{equation}
    g_j
    \asymp
    \frac{g_0}{1+g_0 \beta_0^0 (j \wedge j_m)}.
  \end{equation}
  Therefore
  \begin{equation}
    \epsilon_j^{-1}
    \asymp g_0^{-1/2}+ \gamma (j \wedge j_m)
    \leq g_0^{-1}+ \gamma (j \wedge j_m)
    \asymp g_j^{-1}
  \end{equation}
  as needed.
\end{proof}

\begin{proof}[Proof of Lemma~\ref{lem:summu}]
We recall the definition $\vartheta_j= 2^{-(j-j_m)_+}$ and set $\mu_j = L^{2j} \nu_j$.
By \cite[Proposition~8.3.1]{rg-brief}, the sequence $\mu_j$ satisfies
$\mu_j= O(\vartheta_j g_j)$ and
\begin{equation}
  \mu_{j+1} = L^2 ((1-\gamma\beta_j g_j)\mu_j + \eta_j g_j) + O(\vartheta_jg_j^2),
\end{equation}
where $\eta_j=O(1)$.
Let $\eta_{\geq j} = \sum_{k=j}^\infty L^{-2(k-j)} \eta_k$ and $\hat\mu_j = \mu_j+\eta_{\geq j}g_j$. Then
$\hat\mu_j = O(g_j)$ and
\begin{align}
  \hat\mu_{j+1}
  &= L^2((1-\gamma\beta_{j}g_j)\mu_j + \eta_j g_j) + \pa{\sum_{k=j+1}^\infty L^{-2(k-j-1)} \eta_k} g_j + O(\vartheta_j g_j^2)
  \nnb
  &= L^2\pa{(1-\gamma\beta_{j}g_j)\mu_j + \pa{\eta_j +\sum_{k=j+1}^\infty L^{-2(k-j)} \eta_k}g_j} + O(\vartheta_j g_j^2)
  \nnb
  &= L^2\pa{(1-\gamma\beta_{j}g_j)\mu_j + \pa{\sum_{k=j}^\infty L^{-2(k-j)} \eta_k}g_j} + O(\vartheta_j g_j^2)
  \nnb
  &= L^2(1-\gamma\beta_j g_j) \hat\mu_j + O(\vartheta_j g_j^2).
\end{align}
Iterating this equation together with the boundedness of $\hat\mu_j$ implies
\begin{align}
  \hat\mu_j = L^{-2}\hat\mu_{j+1} + O(\vartheta_j g_j^2)
  = \sum_{l=j}^\infty L^{-2(l-j)} O(\vartheta_l g_l^2) = O(\vartheta_j g_j^2).
\end{align}
We will repeatedly use that (see for example \cite[Exercise~6.1.4]{rg-brief})
\begin{equation}
  \sum_{k=0}^\infty \vartheta_k g_k^2 = O(g_0).
\end{equation}
Hence
\begin{align}
  -\sum_{k=0}^j \mu_k
  = -\sum_{k=0}^j \hat\mu_k + \sum_{k=0}^j \eta_{\geq k} g_k
  = \sum_{k=0}^\infty \eta_{\geq k} g_k + O(g_0).
\end{align}
We now bound the sum on the last right-hand side. By definition and rearranging sums,
\begin{align}
  \sum_{k=0}^{\infty} \eta_{\geq k}g_k
  &= \sum_{k=0}^{\infty} \sum_{l=k}^\infty L^{-2(l-k)} \eta_l g_k
  \nnb
  &=
  \sum_{k=0}^{\infty} \sum_{l=k}^\infty L^{-2(l-k)} \eta_l g_l
  +
    \sum_{k=0}^{\infty} \sum_{l=k}^\infty L^{-2(l-k)} \eta_l \sum_{m=k}^{l-1} (g_m-g_{m+1})
  \nnb
  &\leq
  \sum_{l=0}^\infty (1-L^{-2})^{-1} \eta_l g_l
  +
  \sum_{k=0}^{\infty} \sum_{l=k}^\infty L^{-2(l-k)} \eta_l \sum_{m=k}^{l-1} (\beta_m g_m^2+O(g_m^3))
      .
\end{align}
The last sum can be rearranged and bounded as
\begin{equation}
  \sum_{m=0}^\infty \sum_{k=0}^{m} \sum_{l=m+1}^\infty L^{-2(l-k)} \eta_l (\beta_m g_m^2+O(g_m^3))
  =
  \sum_{m=0}^\infty O(L^{2m})O(L^{-2m}) O(\vartheta_m g_m^2)
  =O(g_0).
\end{equation}
By \cite[(5.3.10) and (5.3.7)]{rg-brief},
there exists a constant $\delta$ independent of $n$ such that
$(1-L^{-2})^{-1}\eta_l \leq \delta \frac{n+2}{n+8} \beta_l$.
Therefore
\begin{equation}
  \sum_{l=0}^\infty (1-L^{-2})^{-1} \eta_l g_l
  \leq \delta \frac{n+2}{n+8} \sum_{l=0}^\infty \beta_l g_{l} + O(g_0)
  \leq \delta \frac{n+2}{n+8} |\log g_{j_m}| + O(1),
\end{equation}
where the last inequality again follows from \cite[Exercise~6.1.4]{rg-brief}.
This concludes the proof of the first inequality in \eqref{e:summu}.
The second inequality is immediate from $\mu_j = O(\vartheta_j g_j)$ and \cite[Exercise~6.1.4]{rg-brief}.
\end{proof}

\section{Spectral gap inequality for single-spin Discrete Gaussian measure}
\label{app:singlespin-dg}

In this appendix, we prove Lemma~\ref{lem:gap-dg-single}.
Thus we prove that for any $\beta>0$, there exists a constant $C_\beta$ such that
the measure $\mu_\psi$ on $\Z$ defined in \eqref{e:mupsi} has a spectral gap uniformly in $\psi \in \R$,
\begin{equation} 
  \label{e:gap-dg-single-app}
  \var_{\mu_\psi}(F(n)) \leq C_\beta 
  \E_{\mu_\psi} \Big( (F(n+2\pi)-F(n))^2 + (F(n-2\pi)-F(n))^2  \Big)   .
\end{equation}

\begin{proof}
It is enough to consider $\psi \in [0, 2\pi]$.
To simplify notation, we assume in this proof that $ \mu_\psi $ is supported on $\Z$
up to rescaling $n$ by a factor $2 \pi$,
i.e.,
\begin{equation} 
\label{eq: new measure mu}
\mu_\psi (n) = e^{V(\psi)} e^{- 2 \pi^2 \beta( n- \frac{\psi}{2 \pi} )^2} .
\end{equation}

We are going to apply the \emph{path method} to evaluate the gap \cite{MR2466937}.
Thus we write
\begin{align}
\var_{\mu_\psi}(F(n))
& = \frac12 \sum_{n,m \in \Z} \mu_\psi ( n) \mu_\psi (m) ( F( n) - F( m) )^2 
\nonumber \\
& =  \sum_{n < m} \mu_\psi (  n) \mu_\psi (  m) 
\; \left( \sum_{i =n}^{m-1} F   (i+1)  - F(  i) \right)^2
\nonumber\\
& \leq  \sum_{i \in \Z} \left(  F  (i+1)  - F(  i) \right)^2
\sum_{n \leq i \atop  m \geq i+1} \mu_\psi (  n) \mu_\psi (  m) \; (m - n),
\label{eq: inequality discrete variance}
\end{align}
where we used the Cauchy--Schwarz inequality in the last inequality and Fubini to change the order of 
the summations.
The Dirichlet form in \eqref{e:gap-dg-single} can be rewritten as
\begin{align} 
  \bbD_{\mu_\psi} (F) & : = \E_{\mu_\psi} \Big( (F(n+1)-F(n))^2 + (F(n-1)-F(n))^2  \Big) \nnb
                      & = \sum_{n \in  \Z}  \left( \mu_\psi (n) + \mu_\psi (n+1) \right) (F(n+1)-F(n))^2.
\end{align}
Thus we deduce from \eqref{eq: inequality discrete variance} that
\begin{align}\label{e:DG-singlespin-pf}
\var_{\mu_\psi}(F(n))
\leq
\max_{i \in \bbZ} \left( 
\sum_{n \leq i \atop  m \geq i+1} 
\frac{\mu_\psi (  n) \mu_\psi (  m) \; (m-n)}{ \mu_\psi (  i) + \mu_\psi   (i+1)  }
\right) 
\bbD_{\mu_\psi} (F).
\end{align}
For $i \geq 0$, the maximum can be bounded by
\begin{equation}
\sum_{n \leq i \atop  m \geq i+1} 
      \frac{\mu_\psi ( n) \mu_\psi(m) (m-n)
      }{ \mu_\psi (  i) + \mu_\psi   (i+1)  }
\leq
\sum_{ m \geq i+1} 
      \frac{\mu_\psi (  m)} 
      { \mu_\psi (  i)}
      \left(  m +  \sum_n \mu_\psi (n)  |n|  \right)
\leq
\sum_{ j \geq 1} 
      \frac{\mu_\psi (  i+j)} 
      { \mu_\psi (  i)}
  (i+j+c_\beta)
  ,
\end{equation}
where we used that $c_\beta = \sum_n \mu_\psi (n)  |n|$ is a constant.
From \eqref{eq: new measure mu}, we see that 
the following bound holds uniformly in $\psi \in [0, 2 \pi]$:
\begin{equation}
\forall j \geq 1, \qquad 
\frac{\mu_\psi   (i + j) }{\mu_\psi (  i)} = 
 e^{- 4\pi^2 \beta(i- \frac{\psi}{2 \pi} ) j}
e^{- 2\pi^2  \beta j^2}
\leq e^{- 4\pi^2 \beta(i- 1)}
e^{- 2\pi^2 \beta j^2}.
\end{equation}
Together, the previous two inequalities imply that the maximum over
$i\geq 0$ in \eqref{e:DG-singlespin-pf} is bounded.
The case $ i<0$ can be controlled in the same way. 
This completes the proof of Lemma \ref{lem:gap-dg-single}.
\end{proof}

\section*{Acknowledgements}

We warmly thank Tom Spencer for his contributions to this paper; his input has been crucial.
We also thank David Brydges and Gordon Slade for a number of important discussions
and for careful reading and many helpful comments on a preliminary version of this paper.
Figure~\ref{fig:hier1} is taken from \cite{rg-brief}.
We acknowledge the support of ANR-15-CE40-0020-01 grant LSD.

\bibliography{all}

\begin{thebibliography}{10}

\bibitem{1302.5971}
A.~Abdesselam, A.~Chandra, and G.~Guadagni.
\newblock Rigorous quantum field theory functional integrals over the p-adics
  {I}: anomalous dimensions.
\newblock Preprint, arXiv:1302.5971.

\bibitem{MR2213477}
D.~Bakry.
\newblock Functional inequalities for {M}arkov semigroups.
\newblock In {\em Probability measures on groups: recent directions and
  trends}, pages 91--147. Tata Inst. Fund. Res., Mumbai, 2006.

\bibitem{MR889476}
D.~Bakry and M.~\'Emery.
\newblock Diffusions hypercontractives.
\newblock In {\em S\'eminaire de probabilit\'es, {XIX}, 1983/84}, volume 1123
  of {\em Lecture Notes in Math.}, pages 177--206. Springer, Berlin, 1985.

\bibitem{MR3926125}
R.~Bauerschmidt and T.~Bodineau.
\newblock A very simple proof of the {LSI} for high temperature spin systems.
\newblock {\em J. Funct. Anal.}, 276(8):2582--2588, 2019.

\bibitem{MR3269689}
R.~Bauerschmidt, D.C. Brydges, and G.~Slade.
\newblock Scaling limits and critical behaviour of the 4-dimensional
  {$n$}-component {$\vert \phi\vert ^4$} spin model.
\newblock {\em J. Stat. Phys.}, 157(4-5):692--742, 2014.

\bibitem{MR3345374}
R.~Bauerschmidt, D.C. Brydges, and G.~Slade.
\newblock Critical two-point function of the 4-di\-men\-sion\-al weakly
  self-avoiding walk.
\newblock {\em Comm. Math. Phys.}, 338(1):169--193, 2015.

\bibitem{MR3339164}
R.~Bauerschmidt, D.C. Brydges, and G.~Slade.
\newblock Logarithmic correction for the susceptibility of the 4-dimensional
  weakly self-avoiding walk: a renormalisation group analysis.
\newblock {\em Comm. Math. Phys.}, 337(2):817--877, 2015.

\bibitem{MR3332940}
R.~Bauerschmidt, D.C. Brydges, and G.~Slade.
\newblock A renormalisation group method. {III}. {P}erturbative analysis.
\newblock {\em J. Stat. Phys.}, 159(3):492--529, 2015.

\bibitem{rg-brief}
R.~Bauerschmidt, D.C. Brydges, and G.~Slade.
\newblock {\em Introduction to a renormalisation group method}.
\newblock Lecture Notes in Math. Springer, to appear.
\newblock Preprint available at \url{http://www.statslab.cam.ac.uk/~rb812/}.

\bibitem{MR3422923}
G.~Benfatto, G.~Gallavotti, and I.~Jauslin.
\newblock Kondo effect in a fermionic hierarchical model.
\newblock {\em J. Stat. Phys.}, 161(5):1203--1230, 2015.

\bibitem{MR1552598}
P.M. Bleher and J.G. Sinai.
\newblock Investigation of the critical point in models of the type of
  {D}yson's hierarchical models.
\newblock {\em Comm. Math. Phys.}, 33(1):23--42, 1973.

\bibitem{MR1704666}
T.~Bodineau and B.~Helffer.
\newblock The log-{S}obolev inequality for unbounded spin systems.
\newblock {\em J. Funct. Anal.}, 166(1):168--178, 1999.

\bibitem{MR1764741}
T.~Bodineau and B.~Helffer.
\newblock Correlations, spectral gap and log-{S}obolev inequalities for
  unbounded spins systems.
\newblock In {\em Differential equations and mathematical physics
  ({B}irmingham, {AL}, 1999)}, volume~16 of {\em AMS/IP Stud. Adv. Math.},
  pages 51--66. Amer. Math. Soc., Providence, RI, 2000.

\bibitem{MR0450480}
H.J. Brascamp and E.H. Lieb.
\newblock On extensions of the {B}runn-{M}inkowski and {P}r\'ekopa-{L}eindler
  theorems, including inequalities for log concave functions, and with an
  application to the diffusion equation.
\newblock {\em J. Functional Analysis}, 22(4):366--389, 1976.

\bibitem{MR1143413}
D.~Brydges, S.N. Evans, and J.Z. Imbrie.
\newblock Self-avoiding walk on a hierarchical lattice in four dimensions.
\newblock {\em Ann. Probab.}, 20(1):82--124, 1992.

\bibitem{MR2523458}
D.C. Brydges.
\newblock Lectures on the renormalisation group.
\newblock In {\em Statistical mechanics}, volume~16 of {\em IAS/Park City Math.
  Ser.}, pages 7--93. Amer. Math. Soc., Providence, RI, 2009.

\bibitem{MR3332938}
D.C. Brydges and G.~Slade.
\newblock A renormalisation group method. {I}. {G}aussian integration and
  normed algebras.
\newblock {\em J. Stat. Phys.}, 159(3):421--460, 2015.

\bibitem{MR3332939}
D.C. Brydges and G.~Slade.
\newblock A renormalisation group method. {II}. {A}pproximation by local
  polynomials.
\newblock {\em J. Stat. Phys.}, 159(3):461--491, 2015.

\bibitem{MR3332941}
D.C. Brydges and G.~Slade.
\newblock A renormalisation group method. {IV}. {S}tability analysis.
\newblock {\em J. Stat. Phys.}, 159(3):530--588, 2015.

\bibitem{MR3332942}
D.C. Brydges and G.~Slade.
\newblock A renormalisation group method. {V}. {A} single renormalisation group
  step.
\newblock {\em J. Stat. Phys.}, 159(3):589--667, 2015.

\bibitem{MR2663712}
P.~Caputo, F.~Martinelli, F.~Simenhaus, and F.L. Toninelli.
\newblock ``{Z}ero'' temperature stochastic 3{D} {I}sing model and dimer
  covering fluctuations: a first step towards interface mean curvature motion.
\newblock {\em Comm. Pure Appl. Math.}, 64(6):778--831, 2011.

\bibitem{MR3844472}
N.~Crawford and W.~De~Roeck.
\newblock Stability of the uniqueness regime for ferromagnetic {G}lauber
  dynamics under don-reversible perturbations.
\newblock {\em Ann. Henri Poincar\'{e}}, 19(9):2651--2671, 2018.

\bibitem{MR1101688}
J.~Dimock and T.R. Hurd.
\newblock A renormalization group analysis of the {K}osterlitz-{T}houless
  phase.
\newblock {\em Comm. Math. Phys.}, 137(2):263--287, 1991.

\bibitem{MR1240586}
J.~Dimock and T.R. Hurd.
\newblock Construction of the two-dimensional sine-{G}ordon model for
  {$\beta<8\pi$}.
\newblock {\em Comm. Math. Phys.}, 156(3):547--580, 1993.

\bibitem{MR1777310}
J.~Dimock and T.R. Hurd.
\newblock Sine-{G}ordon revisited.
\newblock {\em Ann. Henri Poincar\'e}, 1(3):499--541, 2000.

\bibitem{MR2506768}
J.~Ding, E.~Lubetzky, and Y.~Peres.
\newblock The mixing time evolution of {G}lauber dynamics for the mean-field
  {I}sing model.
\newblock {\em Comm. Math. Phys.}, 289(2):725--764, 2009.

\bibitem{MR2585995}
J.~Ding, E.~Lubetzky, and Y.~Peres.
\newblock Mixing time of critical {I}sing model on trees is polynomial in the
  height.
\newblock {\em Comm. Math. Phys.}, 295(1):161--207, 2010.

\bibitem{MR0436850}
F.J. Dyson.
\newblock Existence of a phase-transition in a one-dimensional {I}sing
  ferromagnet.
\newblock {\em Comm. Math. Phys.}, 12(2):91--107, 1969.

\bibitem{MR2917175}
P.~Falco.
\newblock Kosterlitz-{T}houless transition line for the two dimensional
  {C}oulomb gas.
\newblock {\em Comm. Math. Phys.}, 312(2):559--609, 2012.

\bibitem{1311.2237}
P.~Falco.
\newblock {C}ritical exponents of the two dimensional {C}oulomb gas at the
  {B}erezinskii-{K}osterlitz-{T}houless transition, 2013.
\newblock Preprint, arXiv:1311.2237.

\bibitem{MR882810}
J.~Feldman, J.~Magnen, V.~Rivasseau, and R.~S{\'e}n{\'e}or.
\newblock Construction and {B}orel summability of infrared {$\Phi^4\_4$} by a
  phase space expansion.
\newblock {\em Comm. Math. Phys.}, 109(3):437--480, 1987.

\bibitem{MR634447}
J.~Fr{\"o}hlich and T.~Spencer.
\newblock The {K}osterlitz-{T}houless transition in two-dimensional abelian
  spin systems and the {C}oulomb gas.
\newblock {\em Comm. Math. Phys.}, 81(4):527--602, 1981.

\bibitem{MR693402}
K.~Gaw{\polhk{e}}dzki and A.~Kupiainen.
\newblock Triviality of {$\varphi ^{4}\_{4}$}\ and all that in a hierarchical
  model approximation.
\newblock {\em J. Statist. Phys.}, 29(4):683--698, 1982.

\bibitem{MR790736}
K.~Gaw{\polhk{e}}dzki and A.~Kupiainen.
\newblock Massless lattice {$\varphi^4\_4$} theory: rigorous control of a
  renormalizable asymptotically free model.
\newblock {\em Comm. Math. Phys.}, 99(2):197--252, 1985.

\bibitem{MR880526}
K.~Gaw{\polhk{e}}dzki and A.~Kupiainen.
\newblock Asymptotic freedom beyond perturbation theory.
\newblock In {\em Ph\'enom\`enes critiques, syst\`emes al\'eatoires, th\'eories
  de jauge, {P}art {I}, {II} ({L}es {H}ouches, 1984)}, pages 185--292.
  North-Holland, Amsterdam, 1986.

\bibitem{halperinhohenberg1977}
B.~Halperin and P.~Hohenberg.
\newblock Theory of dynamical critical phenomena.
\newblock {\em Rev. Mod. Phys}, 49:435--479, 1977.

\bibitem{MR892924}
T.~Hara.
\newblock A rigorous control of logarithmic corrections in four-dimensional
  {$\phi^4$} spin systems. {I}. {T}rajectory of effective {H}amiltonians.
\newblock {\em J. Statist. Phys.}, 47(1-2):57--98, 1987.

\bibitem{MR892925}
T.~Hara and H.~Tasaki.
\newblock A rigorous control of logarithmic corrections in four-dimensional
  {$\phi^4$} spin systems. {II}. {C}ritical behavior of susceptibility and
  correlation length.
\newblock {\em J. Statist. Phys.}, 47(1-2):99--121, 1987.

\bibitem{MR1936110}
B.~Helffer.
\newblock {\em Semiclassical analysis, {W}itten {L}aplacians, and statistical
  mechanics}, volume~1 of {\em Series in Partial Differential Equations and
  Applications}.
\newblock World Scientific Publishing Co., Inc., River Edge, NJ, 2002.

\bibitem{MR1257821}
B.~Helffer and J.~Sj{\"o}strand.
\newblock On the correlation for {K}ac-like models in the convex case.
\newblock {\em J. Statist. Phys.}, 74(1-2):349--409, 1994.

\bibitem{MR1837286}
M.~Ledoux.
\newblock Logarithmic {S}obolev inequalities for unbounded spin systems
  revisited.
\newblock In {\em S\'eminaire de {P}robabilit\'es, {XXXV}}, volume 1755 of {\em
  Lecture Notes in Math.}, pages 167--194. Springer, Berlin, 2001.

\bibitem{MR2550363}
D.A. Levin, M.J. Luczak, and Y.~Peres.
\newblock Glauber dynamics for the mean-field {I}sing model: cut-off, critical
  power law, and metastability.
\newblock {\em Probab. Theory Related Fields}, 146(1-2):223--265, 2010.

\bibitem{MR2466937}
D.A. Levin, Y.~Peres, and E.L. Wilmer.
\newblock {\em Markov chains and mixing times}.
\newblock American Mathematical Society, Providence, RI, 2009.
\newblock With a chapter by James G. Propp and David B. Wilson.

\bibitem{MR3017041}
E.~Lubetzky, F.~Martinelli, A.~Sly, and F.L. Toninelli.
\newblock Quasi-polynomial mixing of the 2{D} stochastic {I}sing model with
  ``plus'' boundary up to criticality.
\newblock {\em J. Eur. Math. Soc. (JEMS)}, 15(2):339--386, 2013.

\bibitem{MR2945623}
E.~Lubetzky and A.~Sly.
\newblock Critical {I}sing on the square lattice mixes in polynomial time.
\newblock {\em Comm. Math. Phys.}, 313(3):815--836, 2012.

\bibitem{MR3486171}
E.~Lubetzky and A.~Sly.
\newblock Information percolation and cutoff for the stochastic {I}sing model.
\newblock {\em J. Amer. Math. Soc.}, 29(3):729--774, 2016.

\bibitem{MR1063215}
D.H.U. Marchetti, A.~Klein, and J.F. Perez.
\newblock Power-law falloff in the {K}osterlitz-{T}houless phase of a
  two-dimensional lattice {C}oulomb gas.
\newblock {\em J. Statist. Phys.}, 60(1-2):137--166, 1990.

\bibitem{MR1003504}
D.H.U. Marchetti and J.F. Perez.
\newblock The {K}osterlitz-{T}houless phase transition in two-dimensional
  hierarchical {C}oulomb gases.
\newblock {\em J. Statist. Phys.}, 55(1-2):141--156, 1989.

\bibitem{MR1746301}
F.~Martinelli.
\newblock Lectures on {G}lauber dynamics for discrete spin models.
\newblock In {\em Lectures on probability theory and statistics
  ({S}aint-{F}lour, 1997)}, volume 1717 of {\em Lecture Notes in Math.}, pages
  93--191. Springer, Berlin, 1999.

\bibitem{MR1269387}
F.~Martinelli and E.~Olivieri.
\newblock Approach to equilibrium of {G}lauber dynamics in the one phase
  region. {I}. {T}he attractive case.
\newblock {\em Comm. Math. Phys.}, 161(3):447--486, 1994.

\bibitem{MR1269388}
F.~Martinelli and E.~Olivieri.
\newblock Approach to equilibrium of {G}lauber dynamics in the one phase
  region. {II}. {T}he general case.
\newblock {\em Comm. Math. Phys.}, 161(3):487--514, 1994.

\bibitem{MR3098070}
G.~Menz and F.~Otto.
\newblock Uniform logarithmic {S}obolev inequalities for conservative spin
  systems with super-quadratic single-site potential.
\newblock {\em Ann. Probab.}, 41(3B):2182--2224, 2013.

\bibitem{MR1715549}
N.~Yoshida.
\newblock The log-{S}obolev inequality for weakly coupled lattice fields.
\newblock {\em Probab. Theory Related Fields}, 115(1):1--40, 1999.

\bibitem{MR3526836}
O.~Zeitouni.
\newblock Branching random walks and {G}aussian fields.
\newblock In {\em Probability and statistical physics in {S}t. {P}etersburg},
  volume~91 of {\em Proc. Sympos. Pure Math.}, pages 437--471. Amer. Math.
  Soc., Providence, RI, 2016.

\end{thebibliography}
\bibliographystyle{plain}

\end{document}